\documentclass{llncs} 

\usepackage{graphicx}
\usepackage{amsmath}            
\usepackage{amssymb}            
\usepackage{amsfonts}           
\usepackage{url}                
\usepackage{color}              
\usepackage{multirow}           
\usepackage{cite}               
\usepackage{epstopdf}           
\usepackage{complexity}         
\usepackage{paralist}           
\usepackage{wrapfig}            

\newcommand{\ceil}[1]{\left\lceil #1\right\rceil}
\newcommand{\ang}[1]{\langle #1\rangle}
\newcommand{\ZZ}{\mathbb{Z}}    
\newcommand{\RS}{\mathbb{R}} 
\newcommand{\PRE}{\mathbb{R}^+} 

\newcommand{\ropt}{\rho^{\rm opt}_{\infty}} 
\newcommand{\rpar}{\rho^{\rm par}_{\infty}} 
\DeclareMathOperator{\suchthat}{\hspace{0.35em}| \hspace{0.35em}} 

\newcommand{\initpos}{p^{\vdash}}
\newcommand{\goalpos}{p^{\dashv}}
\newcommand{\inittime}{t^{\vdash}}
\newcommand{\goaltime}{t^{\dashv}}

\newcommand{\speedlimit}{\delta_{max}}

\newcommand{\TRUE}{{\bf true}}
\newcommand{\FALSE}{{\bf false}}

\title{On the Complexity of an Unregulated Traffic Crossing%
\thanks{Supported by the National Science Foundation under grant CCF-1117259 and the Office of Naval Research under grant N00014-08-1-1015.}}

\author{Philip Dasler and David M. Mount}

\institute{Department of Computer Science \\
		University of Maryland \\
		College Park, Maryland 20742 \\
		\texttt{\{daslerpc,mount\}@cs.umd.edu}}

\begin{document}
\maketitle
\begin{abstract}
The steady development of motor vehicle technology will enable cars of the near future to assume an ever increasing role in the decision making and control of the vehicle itself. In the foreseeable future, cars will have the ability to communicate with one another in order to better coordinate their motion. This motivates a number of interesting algorithmic problems. One of the most challenging aspects of traffic coordination involves traffic intersections. In this paper we consider two formulations of a simple and fundamental geometric optimization problem involving coordinating the motion of vehicles through an intersection.

We are given a set of $n$ vehicles in the plane, each modeled as a unit length line segment that moves monotonically, either horizontally or vertically, subject to a maximum speed limit. Each vehicle is described by a start and goal position and a start time and deadline. The question is whether, subject to the speed limit, there exists a collision-free motion plan so that each vehicle travels from its start position to its goal position prior to its deadline.

We present three results. We begin by showing that this problem is $\NP$-complete with a reduction from 3-SAT. Second, we consider a constrained version in which cars traveling horizontally can alter their speeds while cars traveling vertically cannot. We present a simple algorithm that solves this problem in $O(n \log n)$ time.  Finally, we provide a solution to the discrete version of the problem and prove its asymptotic optimality in terms of the maximum delay of a vehicle.
\end{abstract}

\section{Introduction}
As autonomous and semi-autonomous vehicles become more prevalent, there is an emerging interest in algorithms for controlling and coordinating their motions to improve traffic flow. The steady development of motor vehicle technology will enable cars of the near future to assume an ever increasing role in the decision making and control of the vehicle itself. In the foreseeable future, cars will have the ability to communicate with one another in order to better coordinate their motion. This motivates a number of interesting algorithmic problems. One of the most challenging aspects of traffic coordination involves traffic intersections. In this paper we consider two formulations of a simple and fundamental geometric optimization problem involving coordinating the motion of vehicles through an intersection.

Traffic congestion is a complex and pervasive problem with significant economic ramifications. Practical engineering solutions will require consideration of myriad issues, including the physical limitations of vehicle motion and road conditions, the complexities and dynamics of traffic and urban navigation, external issues such as accidents and break-downs, and human factors. We are motivated by the question of whether the field of algorithm design can contribute positively to such solutions. We aim to identify fundamental optimization problems that are simple enough to be analyzed formally, but realistic enough to contribute to the eventual design of actual traffic management systems. In this paper, we focus on a problem, the \emph{traffic crossing problem}, that involves coordinating the motions of a set of vehicles moving through a system of intersections. In urban settings, road intersections are \emph{regulated} by traffic lights or stop/yield signs. Much like an asynchronous semaphore, a traffic light locks the entire intersection preventing cross traffic from entering it, even when there is adequate space to do so. Some studies have proposed a less exclusive approach in which vehicles communicate either with one another or with a local controller that allows vehicles, possibly moving in different directions, to pass through the intersection simultaneously if it can be ascertained (perhaps with a small adjustment in velocities) that the motion is collision-free (see, e.g., \cite{dresner_multiagent_2008}). Even though such systems may be beyond the present-day automotive technology, the approach can be applied to controlling the motion of parcels and vehicles in automated warehouses \cite{wurman_coordinating_2008}.

Prior work on autonomous vehicle control has generally taken a high-level view (e.g., traffic routing \cite{dantzig_truck_1959, clarke_scheduling_1964, solomon_algorithms_1987,yu_multi-agent_2012}) or a low-level view (e.g., control theory, kinematics, etc. \cite{fenton_steering_1976, rajamani_vehicle_2011}). We propose a mid-level view, focusing on the control of vehicles over the course of minutes rather than hours or microseconds, respectively. The work by Fiorini and Shiller on velocity obstacles \cite{fiorini_motion_1998} considers motion coordination in a decentralized context, in which a single agent is attempting to avoid other moving objects. Much closer to our approach is work on \emph{autonomous intersection management} (AIM) \cite{au_motion_2010, carlino_auction-based_2013, dresner_multiagent_2004, dresner_multiagent_2005, dresner_multiagent_2008,%
vanmiddlesworth_replacing_2008}.
This work, however, largely focuses on the application of multi-agent techniques and deals with many real-world issues. As a consequence, formal complexity bounds are not proved. Berger and Klein consider a dynamic motion-panning problem in a similar vein to ours, which is loosely based on the video game \emph{Frogger} \cite{berger_travellers_2010}. Their work is based, at least in part, on the work of Arkin, Mitchell, and Polishchuk \cite{Arkin:2008} in which a group of circular agents must cross a field of polygonal obstacles.  These obstacles are dynamic, but their motion is fixed and known \emph{a priori}.

We consider a simple problem formulation of the traffic crossing problem, but one that we feel captures the essential computational challenges of coordinating crosswise motion through an intersection. Vehicles are modeled as line segments moving monotonically along axis-parallel  lines (traffic lanes) in the plane. Vehicles can alter their speed, subject to a maximum speed limit, but they cannot reverse direction. The objective is to plan the collision-free motion of these segments as they move to their goal positions.

After a formal definition of our traffic crossing problem in Section~\ref{sec:prob_def}, we present three results. First, we show in Section~\ref{sec:hardness} that this problem is $\NP$-complete. (While this is a negative result, it shows that this problem is of a lower complexity class than similar PSPACE-complete motion-planning problems, like sliding-block problems \cite{hearn_sliding_2005}.) Second, in Section~\ref{sec:one_sided_solution} we consider a constrained version in which cars traveling vertically travel at a fixed speed. This variant is motivated by a scenario in which traffic moving in one direction (e.g., a major highway) has priority over crossing traffic (e.g., a small road). We present a simple algorithm that solves this problem in $O(n \log n)$ time.  Finally, we consider the problem in a discrete setting in Section~\ref{sec:discrete_setting}, which simplifies the description of the algorithms while still capturing many of the interesting scheduling elements of the problem.  As part of this consideration, we provide a solution to the problem that limits the maximum delay of any vehicle and prove that this solution is asymptotically optimal.

\vspace*{-8pt}

\section{Problem Definition} \label{sec:prob_def}

The Traffic Crossing Problem is one in which several vehicles must cross an intersection simultaneously.  For a successful crossing, all vehicles must reach the opposite side of the intersection without colliding, and they must do so in a reasonable amount of time. 
This time-based restriction exists to encourage an improvement in efficiency over the traffic light regulated crossing.  Here, a ``reasonable amount of time'' is short enough that the traffic cannot simply take turns crossing the intersection (i.e., using the manner in which a traffic light regulates intersections) but instead forces some amount of simultaneity.  

The Traffic Crossing Problem can be posed either as one of optimization (e.g., how quickly can all the cars get across without colliding) or as a decision problem (e.g., can all vehicles cross, collision-free, within a particular time limit).  Here, it is treated as a decision problem so that its parallels to problems like that of Satisfiability can be more easily illustrated.

Formally, a traffic crossing is defined as a tuple $C = (V, \speedlimit)$.  This tuple is comprised of a set of $n$ vehicles $V$ which exist in $\RS^2$ and a global speed limit $\speedlimit \in \PRE$, where $\PRE$ denotes the set of nonnegative reals.  Each vehicle is modeled as a vertical or horizontal open line segment that moves parallel to its orientation.  Like a car on a road, each vehicle moves monotonically, but its speed may vary between zero and the speed limit.  A vehicle's position is specified by its leading point (relative to its direction).  

Each vehicle $v_i \in V$ is defined as a set of properties, $v_i = \{l_i, \initpos_i, \goalpos_i, \inittime_i, \goaltime_i\}$\footnote{The notational use of $\vdash$ and $\dashv$ set above a variable (e.g., $\alpha^{\vdash}$) represent the beginning and end of a closed interval, respectively (e.g., start and end times).}, defined as follows:
\begin{description}
\item[$l_i$:] The length of the vehicle's line segment.
\item[$\initpos_i$:] The starting position of the vehicle, i.e., the vehicle's position prior to its start time (see below).  The position is defined as a point and represents the leading edge of the vehicle.
\item[$\goalpos_i$:] The goal position of the vehicle. The vehicle is considered to have successfully crossed the intersection if its leading point reaches this position either on or before its deadline (see below). 
\item[$\inittime_i$:] The starting time of the vehicle.  The vehicle may not move prior to this time.
\item[$\goaltime_i$:] The deadline for the vehicle.  This is an absolute point in time by which the vehicle must reach its goal position.
\end{description}

The set $V$ and the global speed limit $\speedlimit$ define the problem and remain invariant throughout.  Our objective is to determine whether there exists a collision-free motion of the vehicles that respects the speed limit and satisfies the goal deadlines.  Such a motion is described by a set of functions, called speed profiles, that define the instantaneous speed of the vehicles at time $t$.

\begin{figure}[htbp]
\centerline{ \includegraphics[scale=.35]{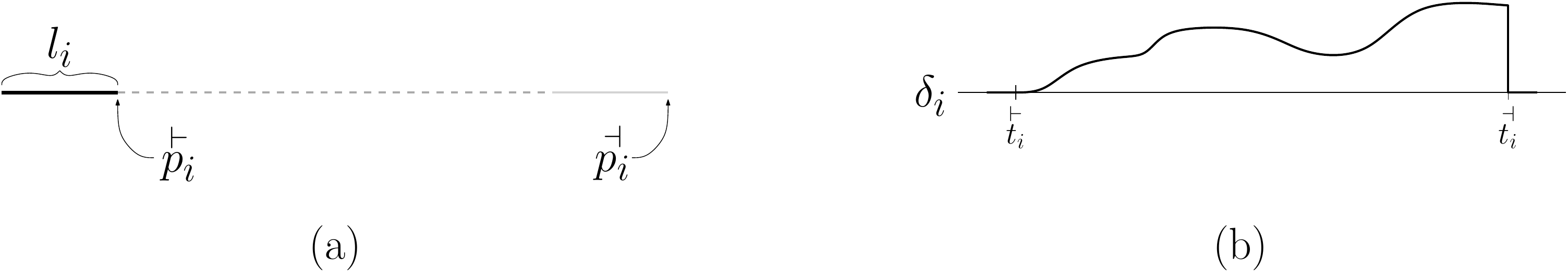}}
\caption{(a) The physical specification of a vehicle $v_i$.  (b) A possible speed profile, $\delta_i$, for a vehicle $v_i$.}
\label{fig:vehicle_definition}
\end{figure}

Formally, this set of functions is defined as $D=\{\delta_i(t) \suchthat i \in[1, n], \forall t, 0 \leq \delta_i(t) \leq \speedlimit\}$.  A vehicle's direction of travel is a unit length vector $d_i$ directed from its initial position to its goal.   Given its speed profile, the position of a vehicle at time $t$ is $p_i(t) =  \initpos_i + d_i\left[\int_0^t \delta_i(x) dx \right]$ and the vehicle $v_i$ inhabits the open line segment between $p_i(t)$ and $p_i(t)-d_i l_i$, which we denote by $\sigma_i(t)$.

\begin{subequations}
A set $D$ of speed profiles is valid if:
\begin{align}
\label{eq:outside_delta} &\forall t \notin \left[\inittime_i, \goaltime_i\right] \hspace{0.5em} \delta_i(t) = 0 \\
\label{eq:inside_delta} &\forall t \in \left[\inittime_i, \goaltime_i\right] \hspace{0.5em} \delta_i(t) \in \left[0, \speedlimit\right] \\
\label{eq:no_intersection} &\forall t \text{ and } \forall v_j, v_i \in V: v_j \neq v_i, \hspace{0.5em} \sigma_i(t) \cap \sigma_j(t) = \emptyset \\
\label{eq:reach_goal} &p_i(\goaltime_i) = \goalpos_i 
\end{align}
\end{subequations}

Equation (\ref{eq:outside_delta}) states that the vehicle may not move either prior to its start time nor after its deadline has passed.  Equation (\ref{eq:inside_delta}) enforces the speed limit and prevents vehicles from traveling in reverse.  Equation (\ref{eq:no_intersection}) prohibits collisions. Equation (\ref{eq:reach_goal}) enforces the goal condition. 

A traffic crossing $C$ is solvable if there exists a valid set of speed profiles $D$.
\vspace*{-5pt}

\section{Hardness of Traffic Crossing}\label{sec:hardness}
Determining whether a given instance of the traffic crossing problem is solvable is $\NP$-complete.  We show its $\NP$-hardness by proving the following theorem:

\begin{theorem}
Given a Boolean formula $F$ in 3-CNF, there exists a traffic crossing $C = (V, \delta)$, computable in polynomial time, such that $F$ is satisfiable if and only if there exists a valid set of speed profiles $D$ for $C$.
\end{theorem}

The input to the reduction is a boolean formula $F$ in 3-CNF (i.e., an instance of 3-SAT).  Let $\{z_1, \ldots, z_n\}$ denote its variables and $\{c_1, \ldots, c_m\}$ denote its clauses. Each variable $z_i$ in $F$ is represented by a pair of vehicles whose motion is constrained to one of two possible states by intersecting their paths with a perpendicular pair of vehicles. This constraining mechanism (seen in Fig.~\ref{fig:copying_values}) is the core concept around which all mechanisms in the reduction are built.  It allows us to represent logical values, to transmit these values throughout the construction, and to check these values for clause satisfaction.

\begin{figure}[htbp]
\vspace*{-10pt}

\centerline{\includegraphics[scale=.35]{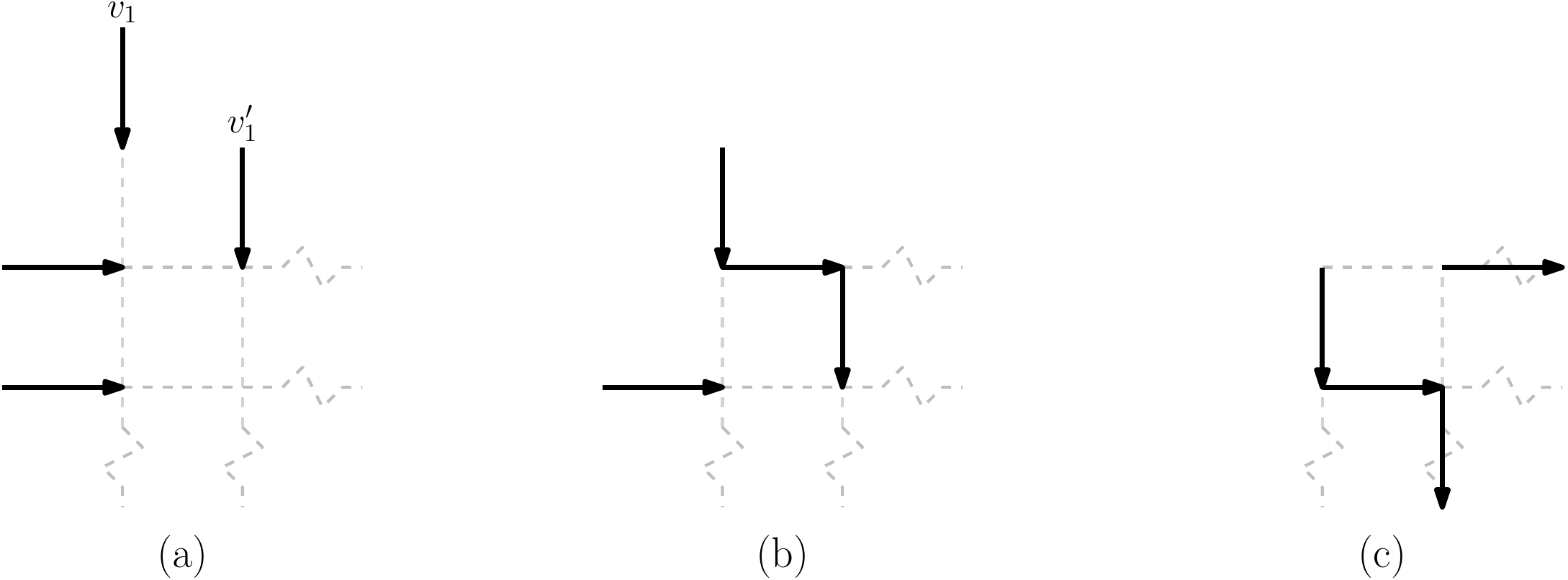}}
\caption{(a) An example of transferring values at $\inittime_i$.  $v_1$ and $v_1'$ are {\TRUE} and {\FALSE}, respectively.  (b) At time $\inittime_i + 1$, the upper horizontal vehicle will take on the value of $v_1'$ while the lower takes the value of $v_1$.  (c) Time $=\inittime_i + 2$. }
\label{fig:copying_values}
\end{figure}

\vspace*{-5pt}

All vehicles in the reduction are of unit length and (barring a few special cases) their deadlines are set so that they can reach their goal position with at most one unit time delay.  More formally, $\goaltime_i - \inittime_i - 1 = \frac{(\|\goalpos_i - \initpos_i\|)}{\speedlimit}$.  In general, the delay may take multiple forms (e.g., the vehicle could take a delay of $1$ at any point during its travel or spread the delay out by traveling slower than $\speedlimit$), but the mechanism described above constrains the delay to only one of two types: a delay of exactly $0$ or $1$ taken immediately at the vehicle's start time. 

For each clause $c_i \in F$, a mechanism is created that forces a collision if, and only if, all three literals are {\FALSE}.  This mechanism checks the positive and negative literals separately, then combines the results in order to determine whether the clause is satisfied.

These mechanisms each require only a constant number of vehicles, resulting in a reduction with a  complexity on the order of $O(n + m)$, where $n$ and $m$ are the number of variables and clauses, respectively. 

What follows is a detailed description of these mechanisms and a formal proof that the Traffic Crossing Problem is $\NP$-complete.
 
\subsection{Variable Representation}
\begin{remark} \label{rem:motion_vis}
As a brief aside, we first begin with an explanation of the conventions used to convey motion over time in the figures throughout this paper.

Since vehicle dynamics and timing differences are difficult to understand in static images, some visualization conventions are used throughout this paper to convey these time dependent properties.  First, delays in a vehicle's movement are visualized by displacing the vehicle by a distance equivalent to the delay.  For example, a vehicle placed $1$ distance unit behind its starting position represents a delay of $1$ time unit (see Fig.~\ref{fig:motion_visualization}(a)).  This positional change is equivalent to a $1$ unit delay as it takes the vehicle this long to reach its original position when traveling at the maximum speed (which the vehicles must do in order to reach their goals in time).  Additionally, we can visualize a vehicle's path in relation to another vehicle traveling perpendicularly by projecting one of the vehicles along the resultant vector of their combined motions (thus the shaded region in Fig.~\ref{fig:motion_visualization}(b)).  The deflection of this band is based on the ratio of the two vehicles' speeds.  If this shaded band intersects both vehicles then a collision will occur between them.  

\begin{figure}[htbp]
\centerline{\includegraphics[scale=.45]{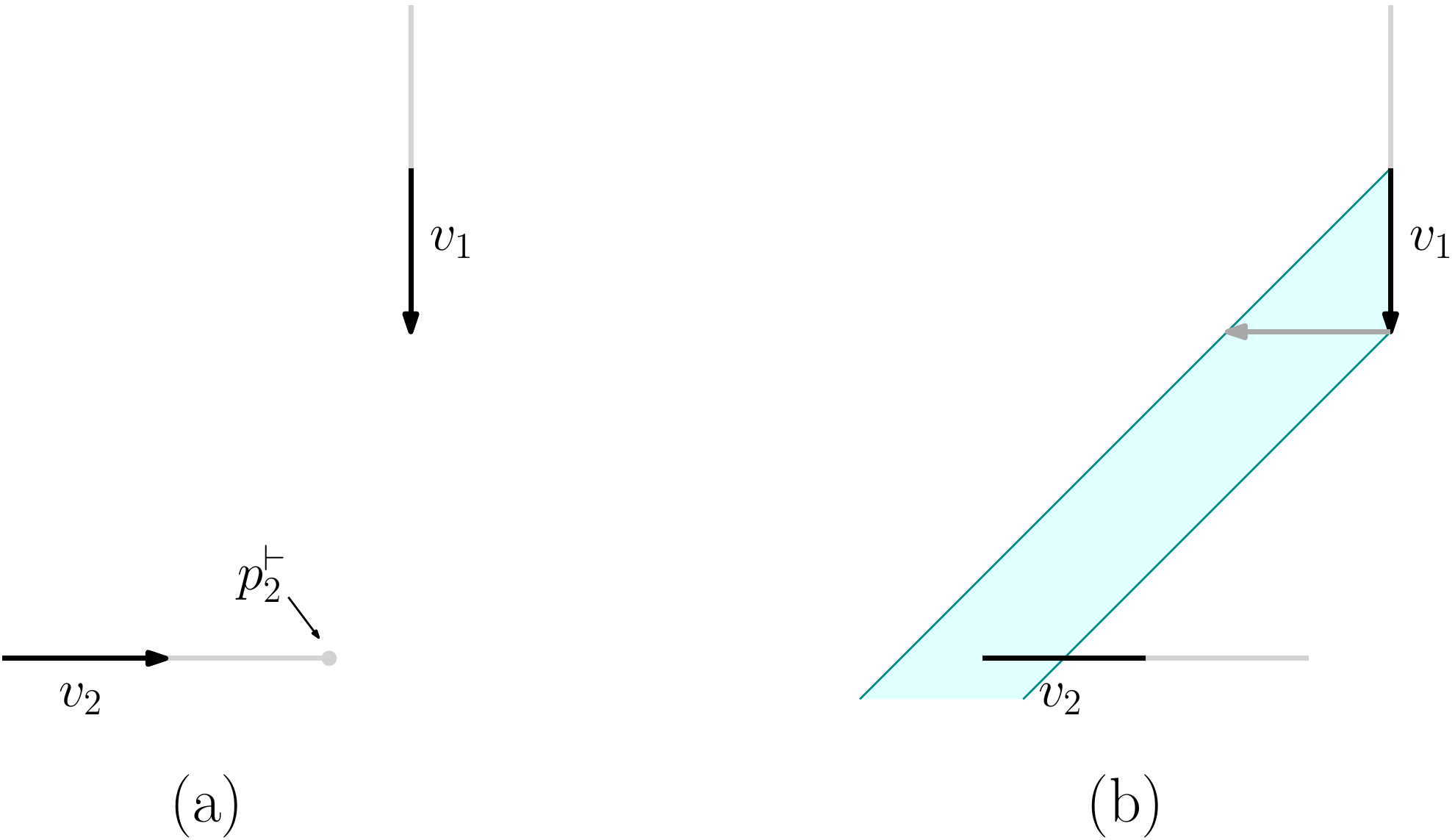}}
\caption{(a) A pair of vehicles traveling toward each other.  Both $v_1$ and $v_2$ have an allowable delay of $1$ time unit, though only $v_2$ has actually done so.  (b) The motion of $v_1$ projected forward in time.  Notice that a collision will occur with $v_2$, which would have been avoided if $v_2$ had not delayed. }
\label{fig:motion_visualization}
\end{figure}

With this understanding of how motion is visualized throughout this paper, we continue on to describe the variable representation gadget.
\end{remark}

Each variable $z_i$ is represented by a pair of vehicles that encode the truth values for both the variable and its negation.  The vehicles in this pair, referred to as \emph{value vehicles}, travel downward in a coordinated manner along two vertical lines that are separated by the unit distance. As mentioned above, a vehicle can endure a delay in the interval $[0, 1]$.  Through the use of additional vehicles, when and how this delay occurs will be constrained.  The vehicles ``carry'' a truth value based on their movement through the system, and doing so requires limiting the vehicles' movements to one of two states. In particular, each value vehicle can either delay for $1$ time unit and then proceed at full speed to the goal or proceed at full speed directly to its final destination, arriving $1$ time unit before its deadline.  These two movement types will be referred to as delay-first ({\TRUE}) and delay-last ({\FALSE}) policies, respectively.  

In order to constrain the delay policies of the value vehicles, two pairs of helper vehicles are added.  The first pair forces one (and only one) of the value vehicles to incur an immediate unit delay.  The second helper pair forces the remaining value vehicle to delay only at the end of its path. 

The helper vehicles in a pair travel together horizontally, are separated vertically by the unit distance, and are placed so that they intersect the value vehicles' paths.  Their goal positions, start times, and end times are all set so that an interaction occurs between them and the value vehicles.

Thus, let $(x, y)$ and $(x+1, y)$ denote the positions of the leading points of a value vehicle pair $(v_1, v_1')$ at time $t$ (see Fig.~\ref{fig:helper_vehicles}).  Place a helper vehicle pair $(u_1', u_1)$ at $(x, y)$ and $(x, y+1)$, respectively, set their goal positions to $(x+2, y)$ and $(x+2, y+1)$, their start times to $t$, and their deadlines to $t+3$.  Similarly, place another helper vehicle pair $(w_1', w_1)$ at $(x, y + \Delta)$ and $(x, y+1+\Delta)$, respectively, set their goal positions to $(x+2, y+\Delta)$ and $(x+2, y+1+\Delta)$, their start times to $t+\Delta$, and their deadlines to $t+3+\Delta$.  The value of $\Delta$ is, essentially, arbitrary and is used here to illustrate that the distance traveled by the value vehicles does not affect their selection of and adherence to one of the two prescribed movement policies.

\begin{lemma} \label{lemma:veh_limit}
Given the pairs $(v_1, v_1')$, $(u_1', u_1)$, and $(w_1', w_1)$ as defined above, the value vehicles $v_1$ and $v_1'$ must each adopt one of the following two movement policies:
\begin{inparaenum}[(a)]
\item delay for exactly $1$ unit of time and then move beyond the paths of $(u_1', u_1)$ at speed $\speedlimit$ (i.e., delay-first); or
\item move beyond the paths of $(w_1', w_1)$ at $\speedlimit$ without delay (i.e., delay-last).
\end{inparaenum}
Additionally, $v_1$ and $v_1'$ may not select the same policy.
\end{lemma}

\begin{proof}
First, notice that because $v_1$ and $u_1'$ are in the same position at time $t$ and are traveling toward each other, they will collide if neither one delays.  Instead, they must choose different movement profiles so that one delays first, allowing the other to pass.  This delay must be exactly $1$ time unit long.  Any longer and the delaying vehicle would miss its deadline; any shorter and there would not be sufficient time for the other vehicle to pass (traveling at $\speedlimit$, a vehicle of length $1$ requires this much time).   

Second, notice that a delay of $v_1'$ necessitates a similar delay of $u_1'$.  This is because it takes $1$ time unit for $u_1'$ to reach the point at which their paths intersect.  If $v_1'$ was to delay $1$ time unit yet $u_1'$ was to leave immediately, they would reach this point simultaneously and collide.  

Given that $u_1'$ must delay if $v_1'$ does and $v_1$ cannot enact the same movement policy that $u_1'$ does, it must be the case that both value vehicles cannot choose to delay for $1$ time unit at this point.  A similar dependency exists between the value vehicles and $u_1$, though this dependency prevents $v_1$ and $v_1'$ from both leaving immediately.

\begin{figure}[htbp]
\centerline{\includegraphics[scale=.45]{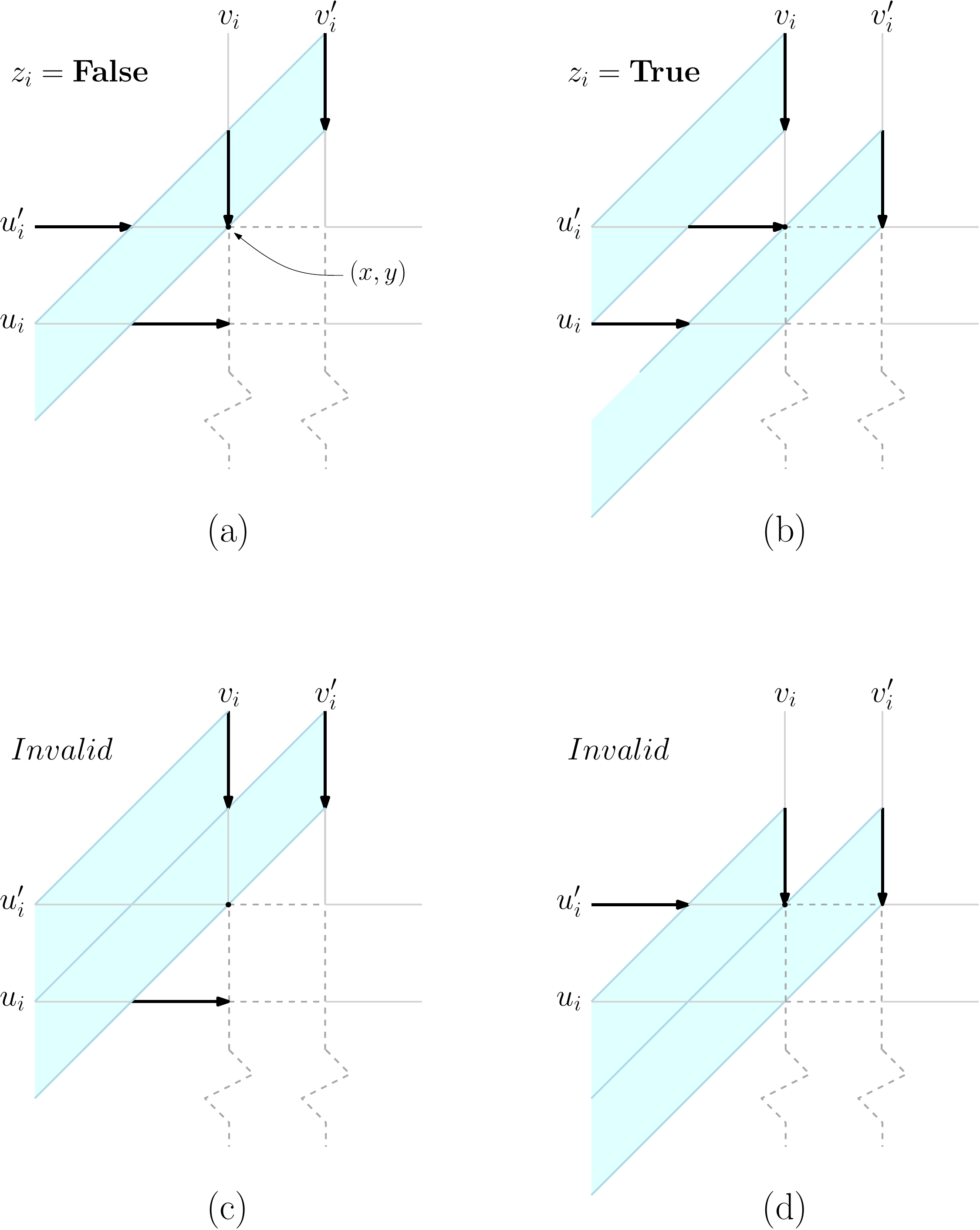}}
\caption{(a/b) Value vehicles taking on opposing values, allowing for valid paths for the helper vehicles.  (c/d) If both value vehicles select the same delay policy, then there is no valid speed profile for one of the helper vehicles. (See Remark~\ref{rem:motion_vis} on figure layouts.)}
\label{fig:helper_vehicles}
\end{figure}

The logic above also holds for the second helper pair, $(w_1', w_1)$, constraining the value vehicles to opposing movement policies until they have moved beyond the paths of the helper vehicles.  This also prevents the value vehicles from swapping movement policies.  To do so would require the lagging vehicle (i.e., the vehicle that adopted the delay-first policy) to speed up while the lead vehicle slows down.  However, given the constraints placed on the vehicles, they are already traveling at the speed limit $\speedlimit$, so the lagging vehicle may not go any faster.
\end{proof}

To represent all of the variables in $\{z_1, \ldots{}, z_n\}$ we create multiple instances of the mechanism described above, one for each variable.  These instances are lined up, one in front of the other, to form a common variable stream (see Fig.~\ref{fig:variable_stream}).  The value vehicles' positions are initialized so that each member in a pair is collinear with the respective members of all other pairs of value vehicles.  Additionally, the starting positions are spaced $s \geq 7$ units apart.  This padding is to allow for the later insertion of a mechanism that regulates the timing of truth values flowing through the system. 

\begin{figure}[htbp]
\centerline{\includegraphics[scale=.38]{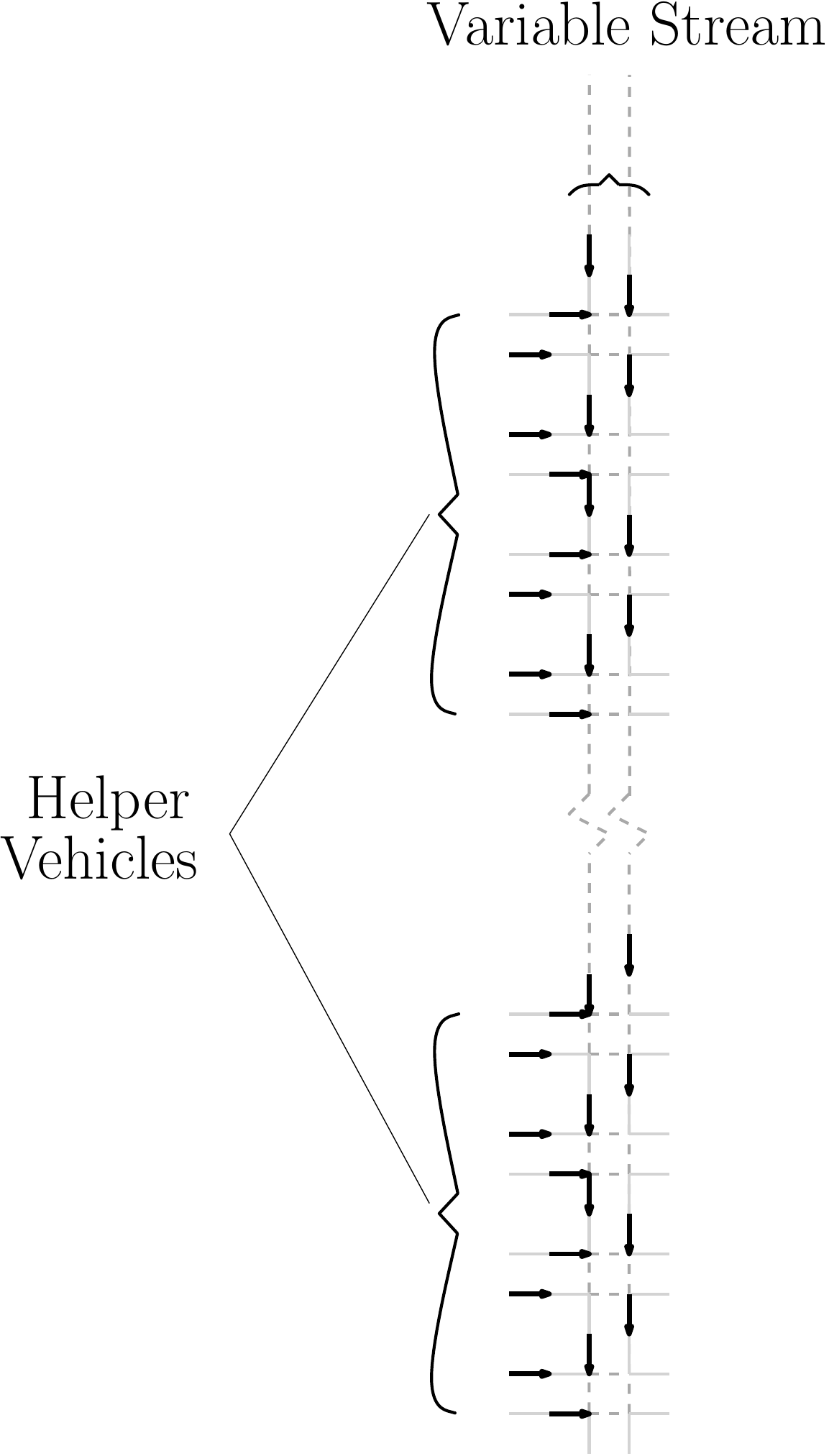}}
\caption{An example of value vehicles arranged into a variable stream representing four variables.}
\label{fig:variable_stream}
\end{figure}

The variable stream is conceptually divided into blocks of length $s|V|$, long enough to accommodate all of the value vehicles and their requisite spacing.  Every clause in $F$ is associated with two of these blocks (one for the positive literals and one for the negative literals), requiring $2|C|$ such blocks (see Fig.~\ref{fig:overview}).  Two extra blocks are added, one at either end of the variable stream, to accommodate the initialization of the value vehicles with the helper vehicles.  Truth values for the appropriate literals will be copied and transferred out of each block to a mechanism which adjusts their relative timing.  This adjustment prepares the vehicles for a final mechanism that validates the satisfaction of the associated clause.  

\begin{figure}[htbp]
\centerline{\includegraphics[scale=.38]{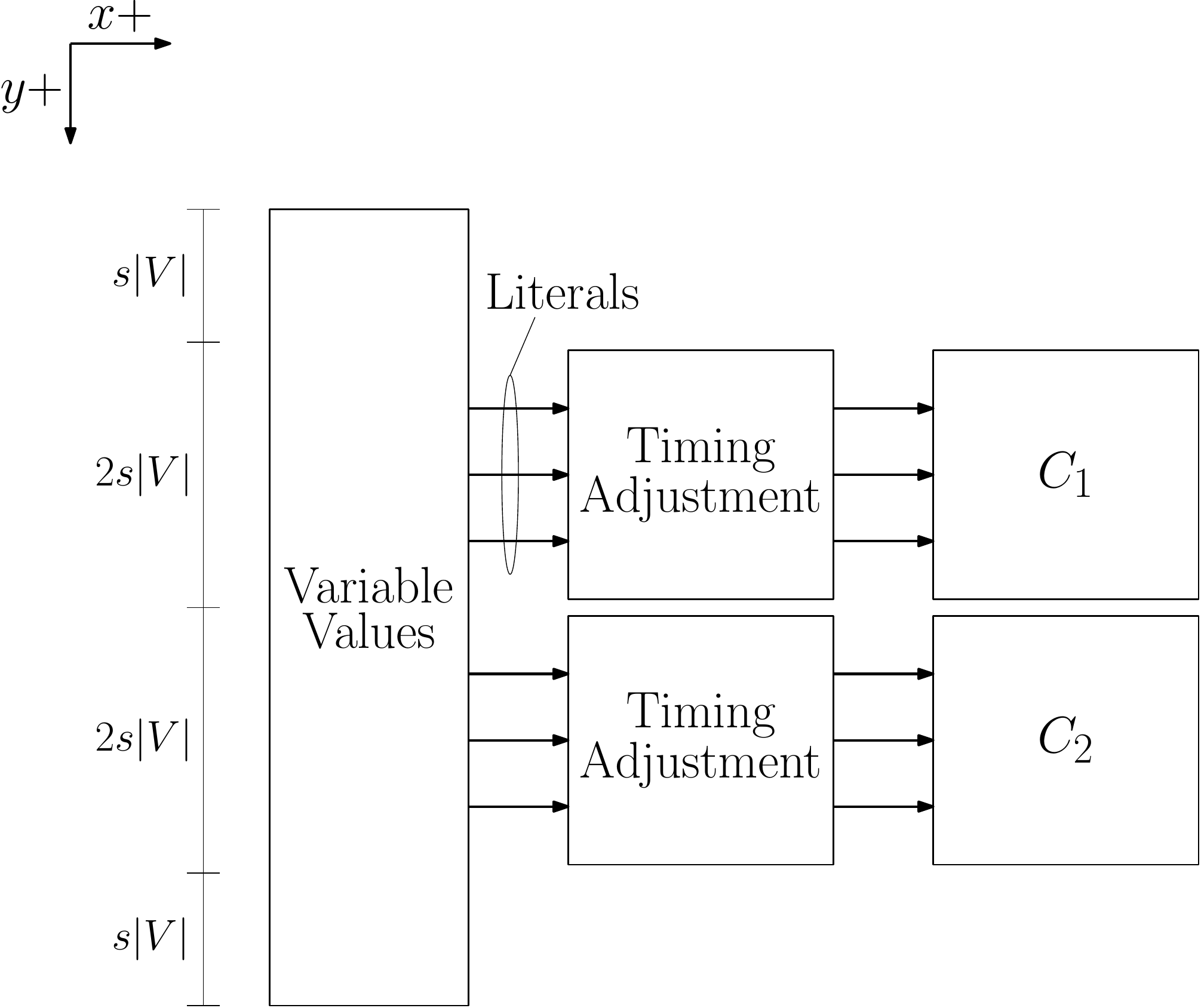}}
\caption{An overview of a reduction from 3-SAT to an instance of the Traffic Crossing Problem.}
\label{fig:overview}
\end{figure}

\begin{subequations}
So, given a formula $F$ with $|C|$ clauses and $|V|$ variables, each variable $z_i$ is represented by a vehicle $v_i$ with the following parameters:
\begin{align}
\initpos_i &=  (0, \space si) \\
\goalpos_i &= (0, \space 2s|V|(|C| + 1) + si) \\
\inittime_i &= 0 \\
\goaltime_i &=  2s|V|(|C|+1) + 1
\end{align}
In addition, the vehicle $v_i'$ is created with similar parameters, but shifted $1$ unit to the right.
\end{subequations}

\subsection{Value Transmission and Timing}

For each clause, the three literal values will need to be carried to the appropriate clause mechanisms so that they arrive in the correct place at the correct time. This requires the introduction of two new mechanisms: one that copies truth values, and one that can adjust the timing of when a value reaches a particular location.  

The first mechanism uses a pair of vehicles whose movement is constrained by a perpendicular pair of vehicles in the same way as the helper vehicles do.  The second mechanism uses a snaking path to induce a delay by increasing the distance traveled. 

\subsubsection{Value Duplication}
In order to perform clause verification we will need the ability to transmit the variable values freely around our space.  To do so, a new pair of parallel vehicles is created, separated by a distance of $1$, whose purpose is to copy these values from the variable stream and carry them elsewhere.  This pair is placed so that its starting position lies on the leftmost side of the variable stream, traveling to the right, and its start time $\inittime_i$ is the time at which the leading edge of the appropriate value pair reaches the vertical position of the uppermost vehicle (see Fig.~\ref{fig:copying_values2}).  Just like the helper vehicles, each of these copy vehicles has their deadlines set so that they may delay for $1$ time unit at most and because of this, the vehicles become a negative copy of the original value vehicles, with the negation on top and the original variable value on the bottom.  We can continue to copy these values in order to carry them through the traffic space, taking orthogonal turns each time we do so.  Any copies along this path that travel vertically will carry the variable's value on the left and the negation on the right.  Any horizontal copy carries the negation on top and the original value below.

\begin{figure}[htbp]
\centerline{\includegraphics[scale=.45]{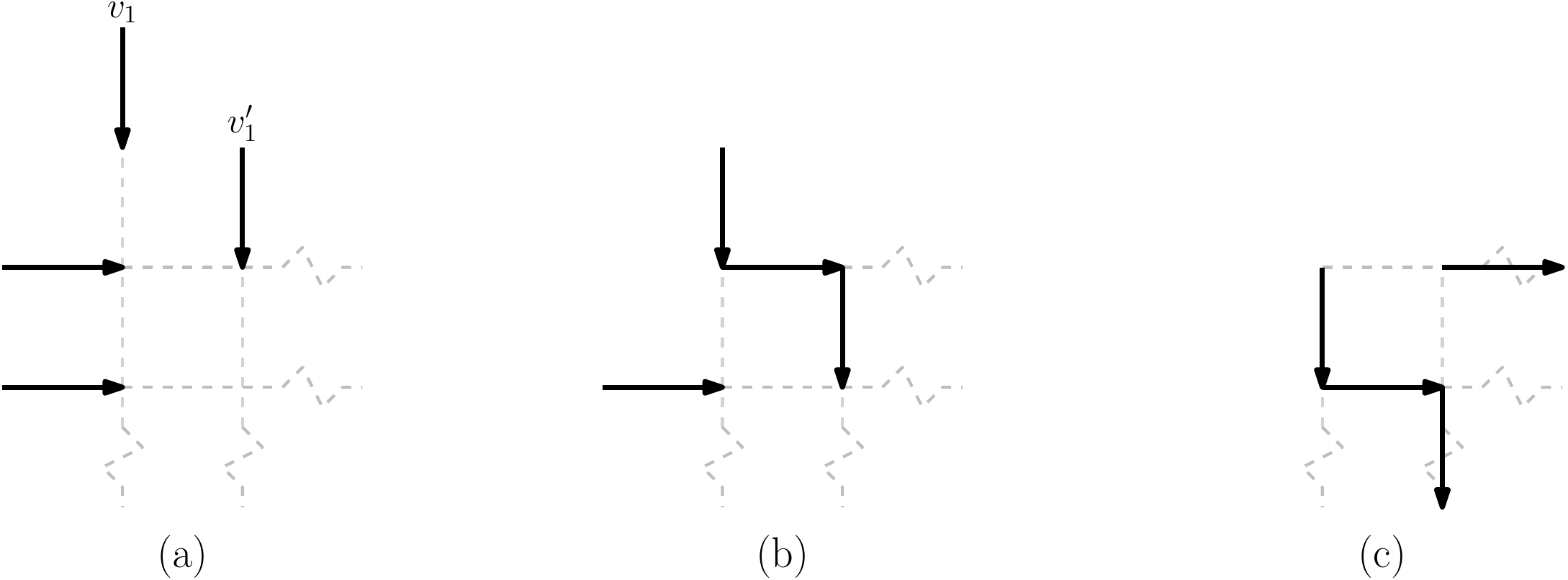}}
\caption{(a) An example of transferring a truth value at start time $\inittime_i$ for the copying vehicles.  In this example, the variable $z_1$ is {\TRUE}, making $v_1$ and $v_1'$ {\TRUE} and {\FALSE}, respectively.  (b) At time $\inittime_i + 1$, notice that in the orthogonal copy the upper vehicle will take on the value of the negation while the lower vehicle takes the original value.  (c) Time $\inittime_i + 2$. }
\label{fig:copying_values2}
\end{figure}

Each of a clause's positive literal values will be copied off of the variable stream simultaneously.  The negative literals are copied similarly.    By chaining vehicle copies across the space we can route the literal values to any location as necessary. 

\subsubsection{Timing and Delays}
The routing of values may require that they travel different distances to reach certain points.  By the structure of our reduction, except when stopped, all vehicles travel at the same speed.  Because of this, any difference in path length will cause a difference in timing that may need to be corrected.  This is done through the introduction of a delay mechanism.  This mechanism is inserted into the path of every copy coming off of the variable stream and can be configured to delay a vehicle pair's leading edge by an arbitrary amount.  This delay does not affect the values carried by the vehicles.  Essentially, the value is routed through an S shape in the mechanism, doubling back on itself (see Fig.~\ref{fig:delay_mechanism}).  The size of this S determines the extra distance that must be traveled and thus the total amount of delay.  A parameter $d$ represents the extra distance added to the S in order to tune the mechanism, leading to a delay of $2d$ (as described below).  Vehicle pairs are arranged in the mechanism as follows, with the first and last referred to as the \emph{incoming pair} and \emph{outgoing pair}, respectively:
\begin{itemize}\itemsep1pt \parskip0pt \parsep0pt
\item $(x, y)$ and $(x, y+1)$ at time $t$,
\item $(x+2+d, y)$ and $(x+3+d, y)$, with a start time of $t+2+d$,
\item $(x+3+d, y+2)$ and $(x+3+d, y+3)$, with a  start time of $t+4+d$,
\item $(x, y+2)$ and $(x+1, y+2)$, with a  start time of $t+6+2d$,
\item $(x, y+4)$ and $(x, y+5)$, with a  start time of $t+8+2d$,
\item $(x+5+d, y+5)$ and $(x+6+d, y+5)$, with a  start time of $t+13+3d$,
\item and $(x+5+d, y)$ and $(x+5+d, y+1)$, with a  start time of $t+17+3d$. 
\end{itemize}

\begin{figure}[htbp]
\centerline{\includegraphics[scale=.45]{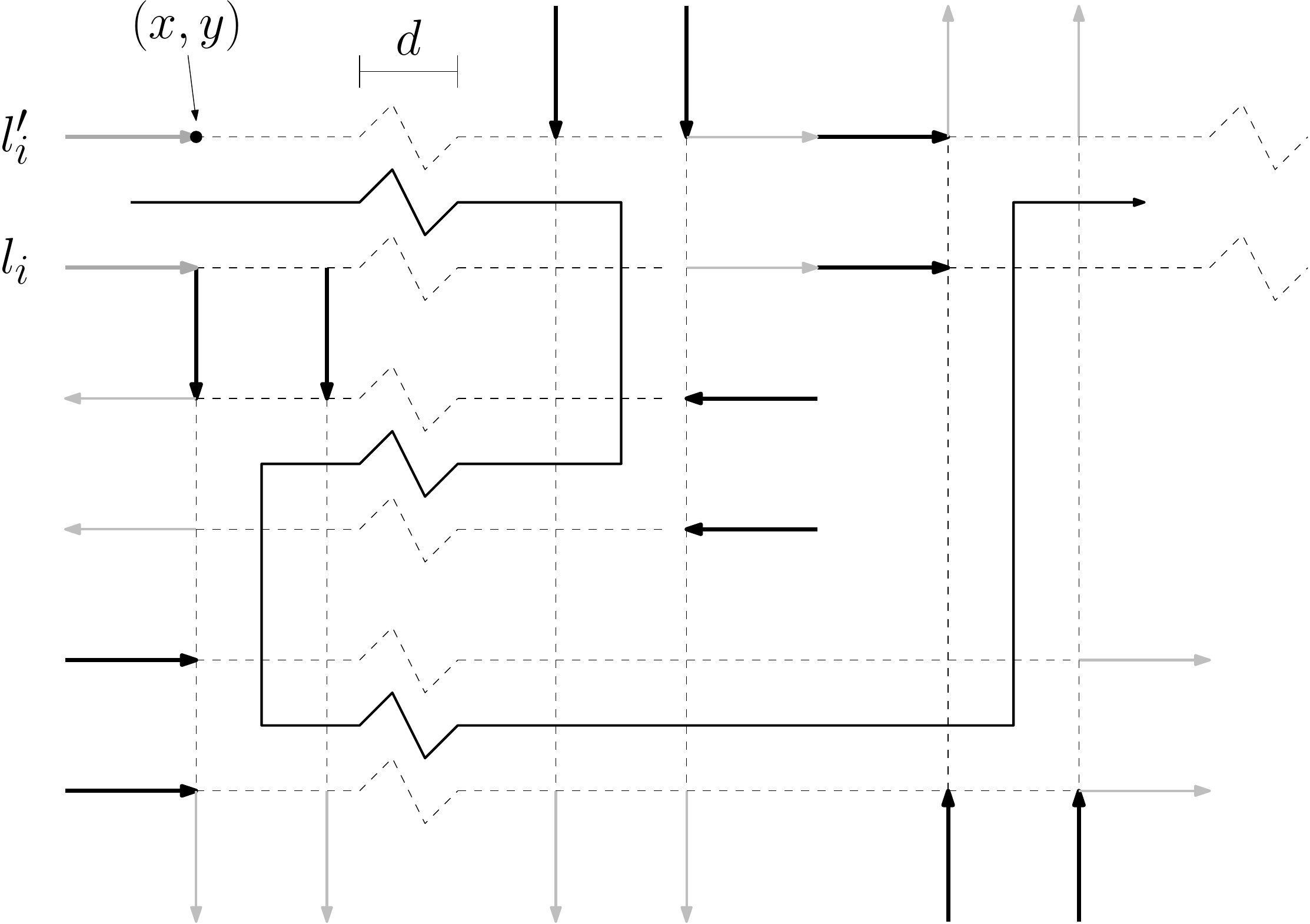}}
\caption{A diagram tracing the path a single pair of truth values $(z_i', z_i)$ take through the delay mechanism.  The dotted vertical line represents where the mechanism can be expanded, separating the vehicles on either side by a distance of $d$ and thus increasing the induced delay.}
\label{fig:delay_mechanism}
\end{figure}

The distance between the incoming vehicle pair and the outgoing vehicle pair is $5 + d$, so, if the incoming pair were to continue on, both pairs would be in the same position at $t+5+d$.  Since the outgoing pair starts at time $t+17+3d$, the mechanism induces a delay in the transmission of the incoming pair of $12 + 2d$.  By adding the delay mechanism to all copies made from the variable stream, we can adjust the relative timing of each vehicle pair by adjusting the value of $d$ in each delay mechanism.

\subsection{Clause Satisfaction}\label{sec:clause_verifier}
Clause satisfaction is verified with a mechanism that forces two particular vehicles to collide if all of the literals are {\FALSE}.  The mechanism consists of two parts: one for the positive literals and one for the negative literals.  Each part contains vehicle pairs representing the literals and their negations, blocking vehicles to appropriately constrain movement, and a verifying vehicle.  If the set of literals do not satisfy the clause, each verifying vehicle is constrained to a single speed profile and they will collide.  First, we look at the half that verifies the negative literals.

Define a point $r = (x, y)$ to be a reference point from which all other positions will be defined at a reference time $t$ (see Fig.~\ref{fig:clause_verifier_3neg}). Next, assume three pairs of incoming vehicles $(l_1', l_1), (l_2', l_2),$ and $(l_3', l_3)$, each a copy of the appropriate variables.  These pairs travel horizontally, $1$ unit apart vertically, with their leading edges $4$ units behind the previous pair.   Thus, the leading edges of the pairs are $(x, y)$ and $(x, y+1)$, $(x-4, y+2)$ and $(x-4, y+3)$, and $(x-8, y+4)$ and $(x-8, y+5)$. 

Next, place two blocking vehicles at $(x-2.5, y+1.5)$ and $(x-6.5, y+3.5)$.  These vehicles have a start time of $t$, travel horizontally to the right, and have their deadlines set so that they must travel at $\speedlimit$ with no delays.  Finally, place a verifying vehicle $v$ at $(x, y)$ with a start time of $t$ and traveling downward.  The deadline for the verifying vehicle is set so that it can delay up to $5$ time units. 

\begin{lemma} \label{lemma:clause_ver}
Given the vehicle pairs, blocking vehicles, and verifying vehicle defined above, the verifier must delay for $5$ time units if all of the literals are {\FALSE} but may delay for less if at least one is {\TRUE}.
\end{lemma}

\begin{proof}
First, notice that every horizontal vehicle in the mechanism is on a possible collision course with the verifying vehicle.  Thus, if the slope of the line between one of these vehicles and the verifying vehicle has a magnitude of $1$ (or if their positions are equal), the vehicle will collide with the verifying vehicle $v$ if both continue without delay.

\begin{figure}[htbp]
\centerline{\includegraphics[scale=.75]{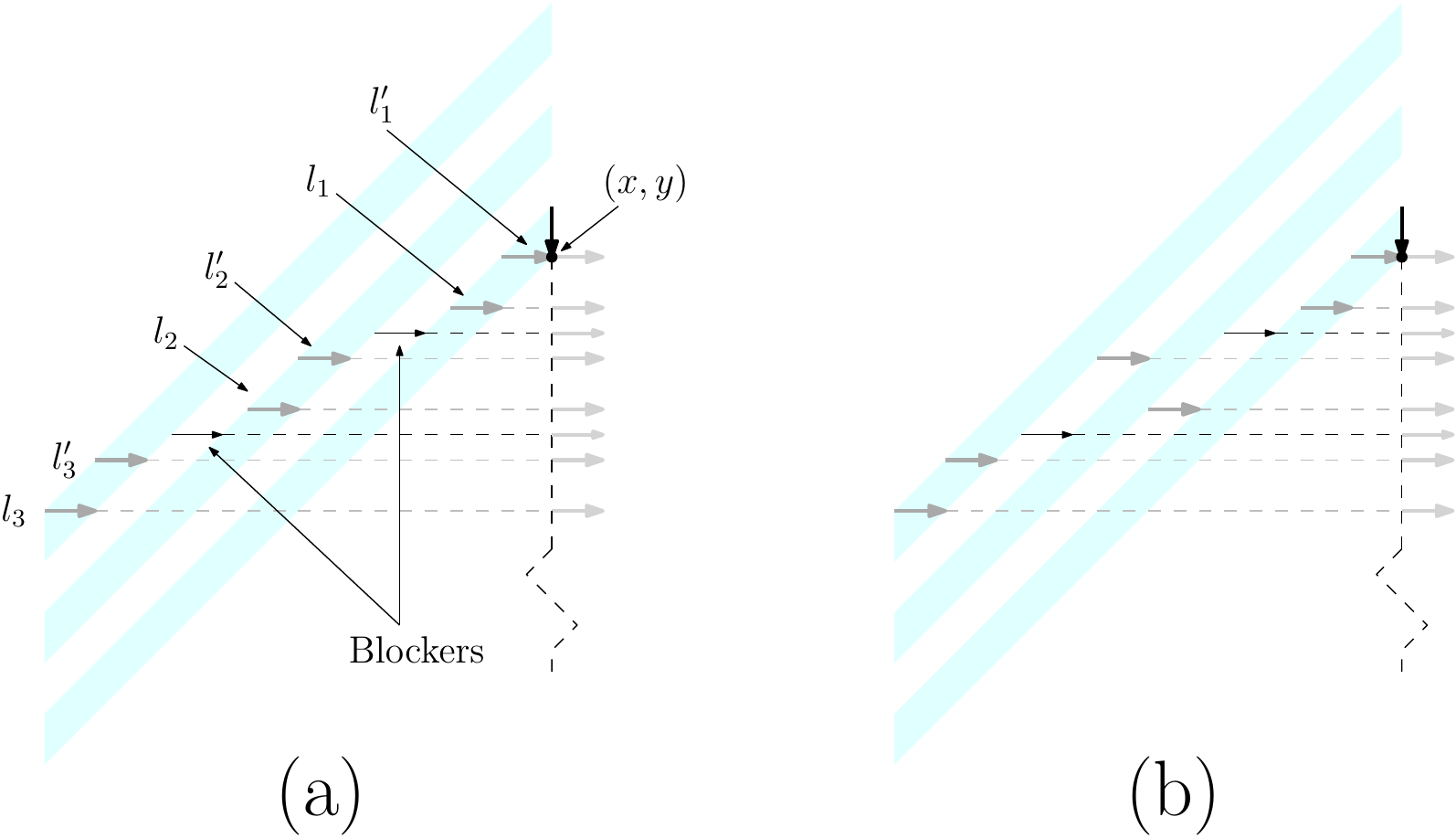}}
\caption{(a) The initialization of the negative half of a clause verifier for the clause $(\neg z_1 \vee \neg z_2 \vee \neg z_3)$ and with each variable $z_i = {\TRUE}$.  (b) The verifier with $z_2= {\FALSE}$ and $z_1 = z_3 = {\TRUE}$. }
\label{fig:clause_verifier_3neg}
\end{figure}

If each variable $z_i$ is {\TRUE} then its negative copy $l_i'$ is {\FALSE}, taking a delay-last movement policy.  This places $l_1'$ at $(x, y)$ and $l_1$ at $(x-1, y-1)$ at time $t$, which would lead to a collision with $v$.  While $l_1'$ could still delay for $1$ unit, $l_1$ no longer has this freedom as it has adopted the delay-first policy.  Thus, to avoid a collision, $v$ must delay for at least $1$.

At time $t+1$, the first blocking vehicle has moved to $(x-1.5, y+1.5)$.  The blocking vehicles' deadlines allow for no delay, so again $v$ must delay.

At time $t+2$, $l_2'$ has moved to $(x-2, y+2)$ and $l_2$ has moved to $(x-3, y+3)$.  Just as with $l_1$, $v$ is forced to delay to avoid a collision.

At time $t+3$, the second blocking vehicle is at $(x-3.5, y+3.5)$, forcing another delay of $v$.

Finally, at time $t+4$, $l_3'$ has moved to $(x-4, y+4)$ and $l_3$ has moved to $(x-5, y+5)$, forcing one last delay of $v$.

Thus, if all of the variables $z_i$ are {\TRUE}, making the negative literals $l_i'$ all {\FALSE}, the verifying vehicle $v$ must delay for $5$ units of time in order to avoid a collision.

If any of the variables are {\FALSE}, their resultant copies $l_i'$ and $l_i$ will have shifted horizontal positions, no longer lying on the line of collision with $v$ (i.e., their slopes are no longer magnitude $1$), allowing $v$ to delay for less than $5$ units and slip between them. 
\end{proof}

The positive half of the mechanism works in the same manner, with slight changes to the incoming literal vehicles and some added vehicles to account for these changes.  First, the incoming literal pairs are not staggered with respect to each other but instead arrive with collinear leading edges and $1$ unit apart (see Fig.~\ref{fig:clause_verifier_3pos}(a)).  Next, a copy of each literal pair is made, traveling downward.  The first copy pair is placed at $(x+5, y)$ and $(x+6, y)$ and has a start time of $t+5$.  The next pair is placed at $(x+3, y+2)$ and $(x+4, y+2)$ with a start time of $t+3$.  The third pair is placed at $(x+1, y+4)$ and $(x+2, y+4)$ and has a start time of $t+1$.  

Next, two blocking vehicles, traveling downward, are added at $(x+2.5, y+9.5)$ and $(x+4.5, y+5.5)$, both with a start time of $t+9$.  

Finally, a verifying vehicle traveling to the right is added at $(x+1, y+12)$, with a start time of $t+9$ and deadline allowing for a delay of up to $5$ time units.  As before, the vehicle will be forced to delay for $5$ time units if the clause is not satisfied by any of the positive literals.

\begin{figure}[htbp]
\centerline{\includegraphics[scale=.75]{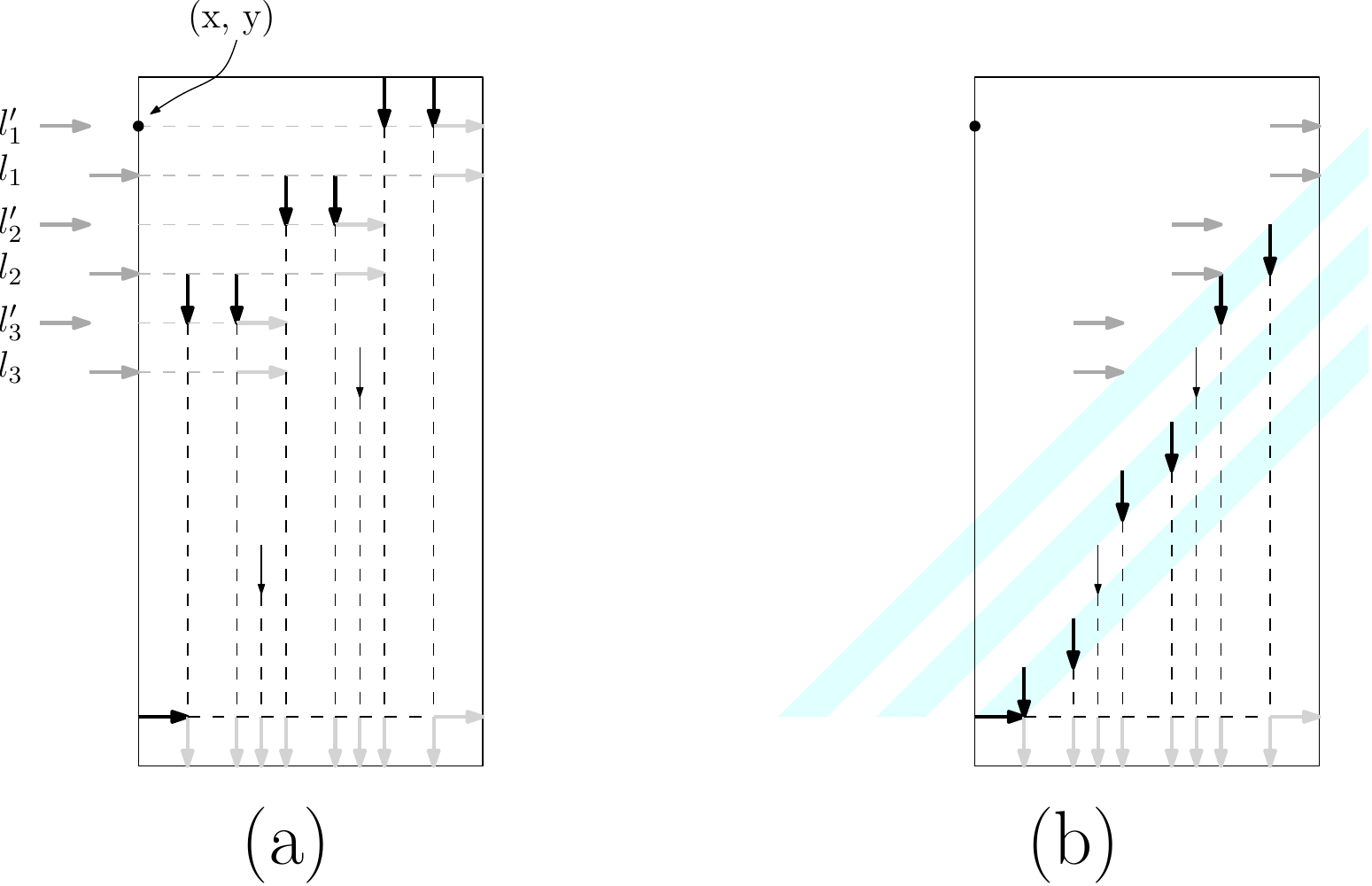}}
\caption{(a) The initialization of the positive half of a clause verifier for the clause $(z_1 \vee z_2 \vee z_3)$ and with each variable $z_i = {\FALSE}$.  (b) The verifier at time $t+9$.  }
\label{fig:clause_verifier_3pos}
\end{figure}

A clause will never have more than three literals, so it will never be the case that both the positive and negative halves of the clause verifier will have three literals.  Blocking vehicles are added to take the place of missing literals in each half and their deadlines are set so that no delay is allowed.  In this way, the verifying vehicles are still forced to delay for $5$ units when their associated set of literals do not satisfy the clause.

The positive and negative halves of the mechanism are placed so that the paths of the verifying vehicles intersect.  However, the time at which each half processes its literals may differ, dependent on which variables are being evaluated and the distance their values must travel to reach the mechanism.  This can be compensated for in the delay mechanisms so that the verifying vehicles will collide with one another if both delay for $5$ time units.  In this way, if a clause is not satisfiable, a collision is inevitable, rendering the traffic crossing unsolvable.  If the clause is satisfiable, one or both of the verifying vehicles will have at least two movement options, allowing them to avoid a collision.

\subsection{Complete System Example} \label{sec:example}
In the complete system, all of the variables are stacked on top of each other to form a variable stream.  The appropriate literals are extracted, passed through a delay mechanism, and routed to their clause verifier halves.  These mechanisms output a vehicle that will have delayed for $5$ time units if the variable assignments do not satisfy their respective clauses.  The verifier vehicles from each clause will collide if neither set of literals satisfies them.  An example of a 3-SAT reduction for the formula $(\neg z_1 \vee z_2 \vee \neg z_3)$ can be seen in Fig.~\ref{fig:example_SAT}.

\begin{figure}[htbp]
\centerline{\includegraphics[scale=.35]{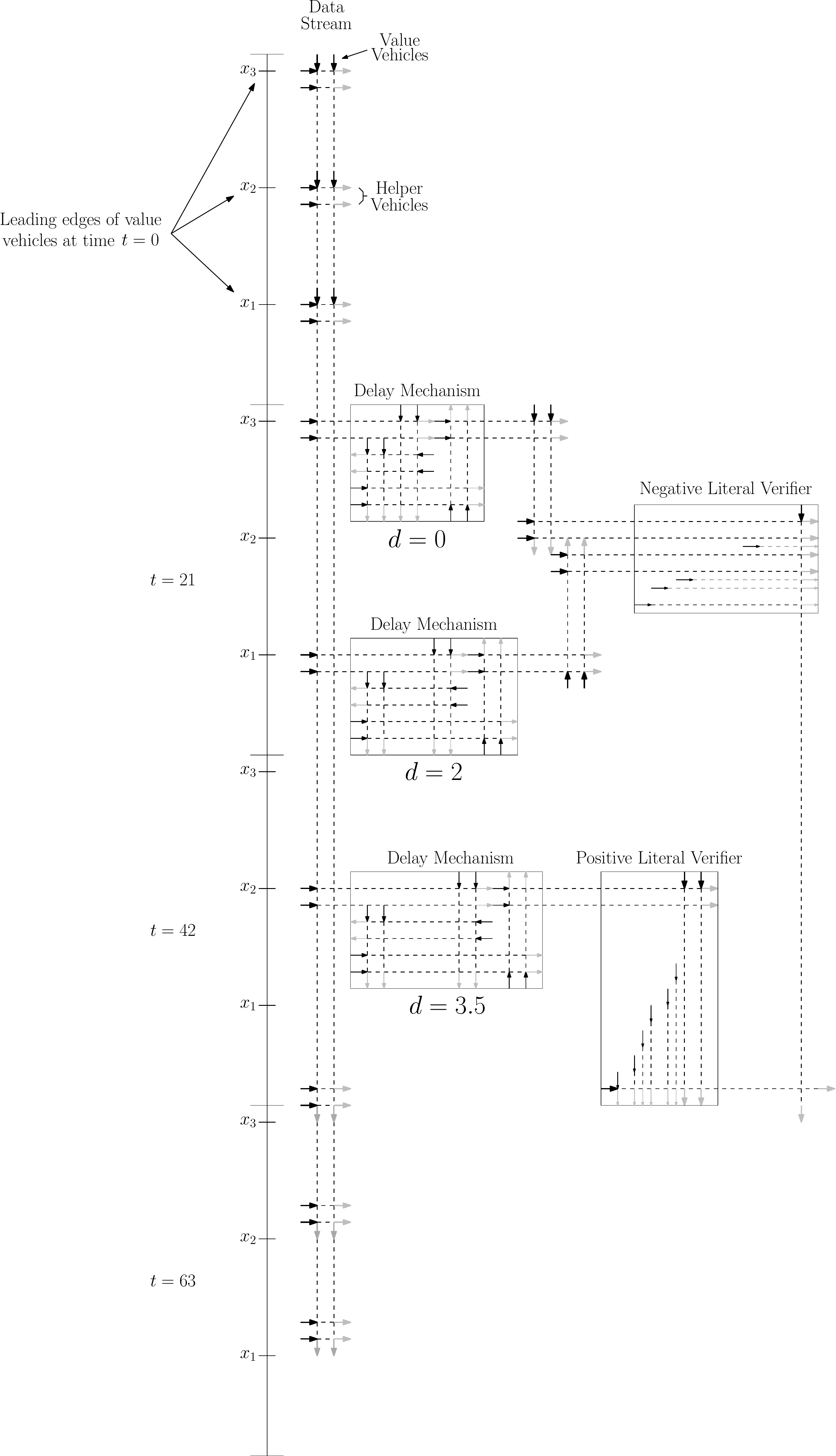}}
\caption{An example of a 3-SAT problem with $F = (\neg z_1 \vee z_2 \vee \neg z_3)$, expressed as a traffic crossing.}
\label{fig:example_SAT}
\end{figure}

\subsection{Analysis of Reduction Complexity}

Every variable in the formula $F$ requires $6n$ vehicles: one for the variable, one for its negation, and the two helper pairs.  Next, when considering each of the $m$ clauses, the greatest number of vehicles is necessary when all of the literals are positive.  $15$ are needed for the positive verifier, $9$ for the negative verifier, $14$ for each of the two delay mechanisms, and $12$ for routing, for a total of at most $64$ vehicles per clause.  The complexity of translation is then $6n + 64m$ and is therefore on the order of $O(n + m)$.

As described above, the constructed mechanisms will only allow for a valid set of speed profiles if the formula $F$ is satisfiable.  Given this and the polynomial time needed to create the reduction, the traffic crossing problem is $\NP$-hard.
\vspace*{-10pt}

\subsection{Membership in $\NP$}

\begin{lemma}
The Traffic Crossing Problem is in $\NP$.
\end{lemma}
Given an input instance of the Traffic Crossing Problem $C = (V, \speedlimit)$, consisting of $n$ vehicles in $V$, where all numeric values are given with $b$ bits of precision, we demonstrate a certificate of size $O(n^2)$ from which it is possible to validate a solution in $O(n^4(b + \log n))$ time.

\begin{description}
\item[Certificate] - For each pair of orthogonal vehicles, $v_i, v_j$, their paths cross at a single intersection.  The certificate provides a priority for each such pair, specifying which vehicle crosses through the intersection first.  This requires $O(n^2)$ bits.
\end{description}

Let $D$ be any valid set of speed profiles.  Let $P(D)$ denote the associated certificate providing vehicle priorities.  To validate $D$, we consider a specific instance called the \emph{full-speed profiles} as follows:
\begin{itemize}
\item Each vehicle $v$ moves at full speed (i.e., $\speedlimit$) until either (1) arriving at an intersection, (2)it is about to collide with the rear end of a stopped vehicle in the same lane, or (3) it has reached its destination.
\item If arriving at an intersection, the vehicle waits until all vehicles which have priority over it, according to $P$, have passed through the intersection.  The vehicle will proceed through the intersection at full speed once the last of these vehicles has passed.
\item If the vehicle has stopped in order to avoid a collision with the vehicle in front of it, it will proceed at full speed once the blocking vehicle has as well.
\item If the vehicle has reached its destination it will stop as, of course, there is no more to be done.
\end{itemize}

We refer to this instance of $D$ as $D_{full}$.  To establish correctness, it suffices to show (1) $D_{full}$ is valid if $D$ is valid and (2) $D_{full}$ can be simulated in $O(poly(n,b))$ time to determine its validity.

\begin{lemma}\label{lemma:dfull_validity}
If $D$ is a valid set of speed profiles then $D_{full}$ is also valid.
\end{lemma}

\begin{proof}
Define a \emph{significant event} (for either profile) to be the time at which some vehicle $v_i$ enters or leaves some intersection $\chi_j$.  These events will be referred to as $t^-(i,j)$ and $t^+(i, j)$, respectively, for the original set of speed profiles $D$.  The significant events for the full speed profile will be denoted as $t^-_{full}(i,j)$ and $t^+_{full}(i, j)$.  We will show:
\begin{enumerate}
\item there are no collisions in $D_{full}$
\item $\forall i,j t^\pm_{full}(i,j) \leq t^\pm(i.j)$
\item arrival times at destinations are earlier than or equal to those in $D$ when following $D_{full}$.
\end{enumerate}

\begin{enumerate}
\item No rear ending can occur by definition of the $D_{full}$ policy.  Also, no intersection collisions can occur between crossing vehicles because (by priority) one is required to wait for the other.
\item We prove this by induction in time over the significant events.  Initially, both profiles are in the same configuration, as given by the problem definition $C$.  Suppose toward contradiction that there exists a significant event concerning vehicle $v_i$ and intersection $\chi_j$ where $t^-_{full}(i,j) > t^-(i,j)$.  Consider the first such event.  There are two possible reasons why $v_i$ did not enter intersection $\chi_j$ at time $t^-(i,j)$ in profile $D_{full}$:
	\begin{enumerate}
	\item It is waiting for some crossing vehicle $v_k$ to exit the intersection (see Fig.~\ref{fig:dfull_faster}(a)).  By definition, $v_k$ must have a higher priority in $D$ (i.e., it passes prior to $v_i$ in $D$), but $v_k$ must have exited the intersection prior to $t^-(i,j)$.  This contradicts the induction hypothesis that $t^-_{full}(i,j) > t^-(i,j)$.
	\item $v_i$ can't proceed because it would rear-end the previous stopped vehicle (see Fig.~\ref{fig:dfull_faster}(b)).  In this situation, there is a chain of $1$ or more vehicles stopped in front of $v_i$, where the first vehicle in the chain $v_{i'}$ is waiting at some intersection $\chi_{j'}$ for some vehicle $v_{k'}$ with priority to pass.  In this case we can apply the argument above to $v_{i'}$, $\chi_{j'}$, and $v_{k'}$.
	Finally, the same holds for each $t^+$ value as they are each equal to the $t^-$ values offset by the constant length of the vehicles.  In other words, for vehicle $v_i$:
	\begin{align*}
	t^+_{full}(i,j) &= t^-_{full}(i,j) + length(v_i)\text{, and}\\
	t^+(i,j) &= t^-(i,j) + length(v_i)
	\end{align*}
	\end{enumerate}
	
\begin{figure}[htbp]
\centerline{\includegraphics[scale=.40]{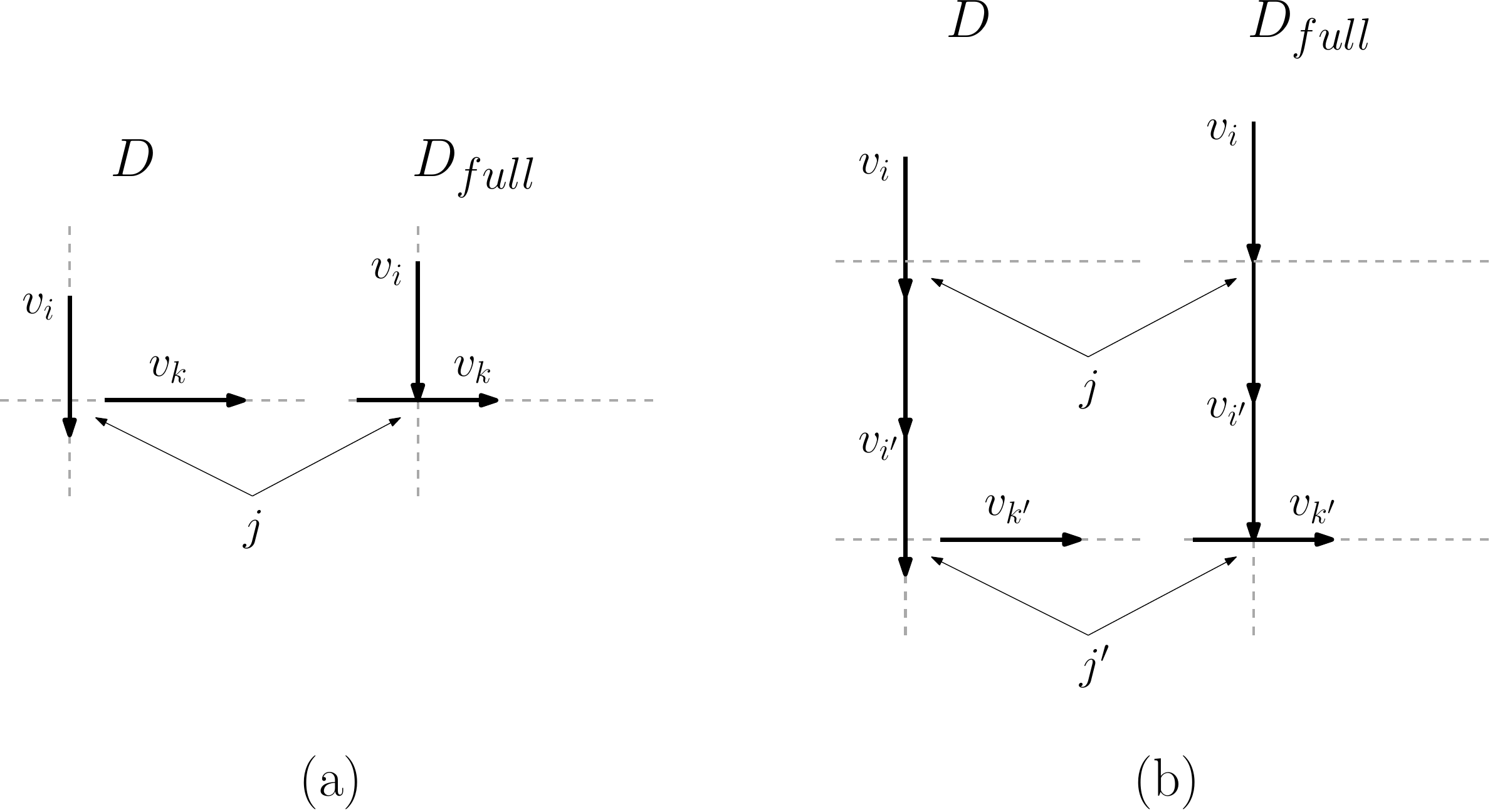}}
\caption{A figure comparing relative event times illustrating the cases in which (a) $t^-_{full}(i,j) > t^-(i,j)$ and (b) a traffic jam occurs.}
\label{fig:dfull_faster}
\end{figure}

\item From argument 2 above, $t^+_{full}(i,j) \leq t^+(i,j)$.  Given this, adding the constant distance from the last intersection to the destination to each value does not change this relationship.  Therefor, vehicles following $D_{full}$ will arrive at a time equal to or earlier than vehicles following $D$.
\end{enumerate}
\end{proof}

\begin{lemma}
Given an instance of the Traffic Crossing Problem of size $n, b$ and certificate $P$, we can simulate $D_{full}$ in $O(n^4(b + \log n))$ time.
\end{lemma}
\begin{proof}
We may assume without loss of generality that the maximum speed is 1 unit per second.  Recall the notion of significant events from Lemma~\ref{lemma:dfull_validity}.  Simulation of the system using $P$ is a simple discrete time event simulation in which we advance from one significant event to the next.  Observe that given suitable data structures, we can process each significant event in $O(n^2)$ time (ignoring numeric issues).  The issue that remains is the number of bits of precision needed to represent the times at which each significant event occurs.  Significant event times can be computed as follows:
\[t^+_{full}(i,j) = t^-_{full}(i,j) + length(v_i)\]

Let $\chi_{j'}$ denote the next intersection along the lane in which vehicle $v_i$ is moving.  The time at which $v_i$ hits intersection $\chi_{j'}$ is:
\[t^+_{full}(i,j) + dist(\chi_j, \chi_{j'})\]

At this time, the vehicle will either continue directly through the intersection, thus implying that the value of $t^-_{full}(i,j')$ is equal to the value above, or it will be forced to wait for some other significant event before it can move.  Observe, then, that each significant event time is the sum of the vehicle length and the distances between consecutive intersections.  If all coordinates are $b$-bits precise, then each time involves an $O(n)$-fold sum of $b$-bit numbers, or $O(b+\log n)$ bits, for a total of $O(n^2(b+\log n))$ bit operations.  Finally, each vehicle can pass through at most $O(n)$ intersections for a total of $O(n^2)$ significant events.  Thus, overall, the number of bit operations is less than or equal to the number of significant events, times the processing time for each, times the number of bits for each or $O(n^2n^2(b+\log n) = O(n^4(b+\log n))$.
\end{proof}

Given the Traffic Crossing Problem's demonstrated hardness and membership in $\NP$, we conclude that the problem is $\NP$-complete.

\vspace*{-10pt}

\section{A Solution to the One-Sided Problem}\label{sec:one_sided_solution}
While the generalized Traffic Crossing Problem is $\NP$-complete, it is possible to solve a constrained version of the problem more efficiently.  The complexity of the generalized Traffic Crossing Problem arises from the interplay between horizontal and vertical vehicles, which results in a complex cascade of constraints. To break this interdependency, the vertically traveling vehicles are given priority, allowing them to continue through the intersection at a fixed speed.  In this variant, called the \emph{one-sided problem}, the horizontal vehicles can plan their motion with complete information and without fear of complex constraint chains.  

First, we assume that the vertically traveling vehicles are invariant and are all traveling at the same speed, $s_n$.  With vertical vehicle motion now fixed, there is no way for horizontal vehicles to affect each other and movement profiles for each can be found in isolation from the others.   Finally, we assume that all vehicles are of length $l$ and in general position.

For the purpose of illustration we begin with a simplified version of the problem and then, over the course of three cases, relax the restrictions until we are left with a solution to the original problem under the fixed, one-sided policy described above.  These three cases are:

\begin{description}
\item[Intersection Between One-Way Highways]  \hspace{0px}
	\begin{itemize}
		\item Vertical vehicles approach from the North only.
		\item Horizontal vehicles approach from the West only.
		\item Each vehicle is in its own lane (i.e., no two vehicles are collinear).
	\end{itemize}
\item[Intersection Between a One-Way Street and a Two-Way Highway] \hspace{0px}
	\begin{itemize}
		\item Vertical vehicles approach from the North and the South.
		\item Horizontal vehicles approach from the West only.
		\item There is a single horizontal lane (i.e., all horizontal vehicles are collinear) and one or more vertical lanes.
	\end{itemize}
\item[Intersection Between Two-Way Highways] \hspace{0px}
	\begin{itemize}
		\item Vertical vehicles approach from the North and the South.
		\item Horizontal vehicles approach from the West and the East.
		\item There are $k$ horizontal lanes, one or more vertical lanes, and vehicles may be collinear.
	\end{itemize}
\end{description}

\vspace*{-15pt}

\subsection{Intersection Between One-Way Highways}
Formally, vehicles from the North are in the subset $N \subset V$ and their direction of travel is $d_n = (0, -1)$, where as vehicles from the West are in the subset $W \subset V$ with a direction of travel of $d_w = (1, 0)$.  Again, our only task is to find valid speed profiles for vehicles coming from the West. 

To begin, the problem space is transformed so that the vehicles in $W$ are represented as points rather than line segments.  This makes movement planning simpler while maintaining the geometric properties of the original space.  Every vehicle in $W$ is contracted from left to right, until it is reduced to its leading point.  In response, the vehicles in $N$ are expanded, transforming each into a square obstacle with sides of length $l$ (see Fig. \ref{fig:space_transform}) and with their left edges coincident with the original line segments.  
	
\begin{figure}[htbp]

\vspace*{-10pt}

\centerline{\includegraphics[scale=.25]{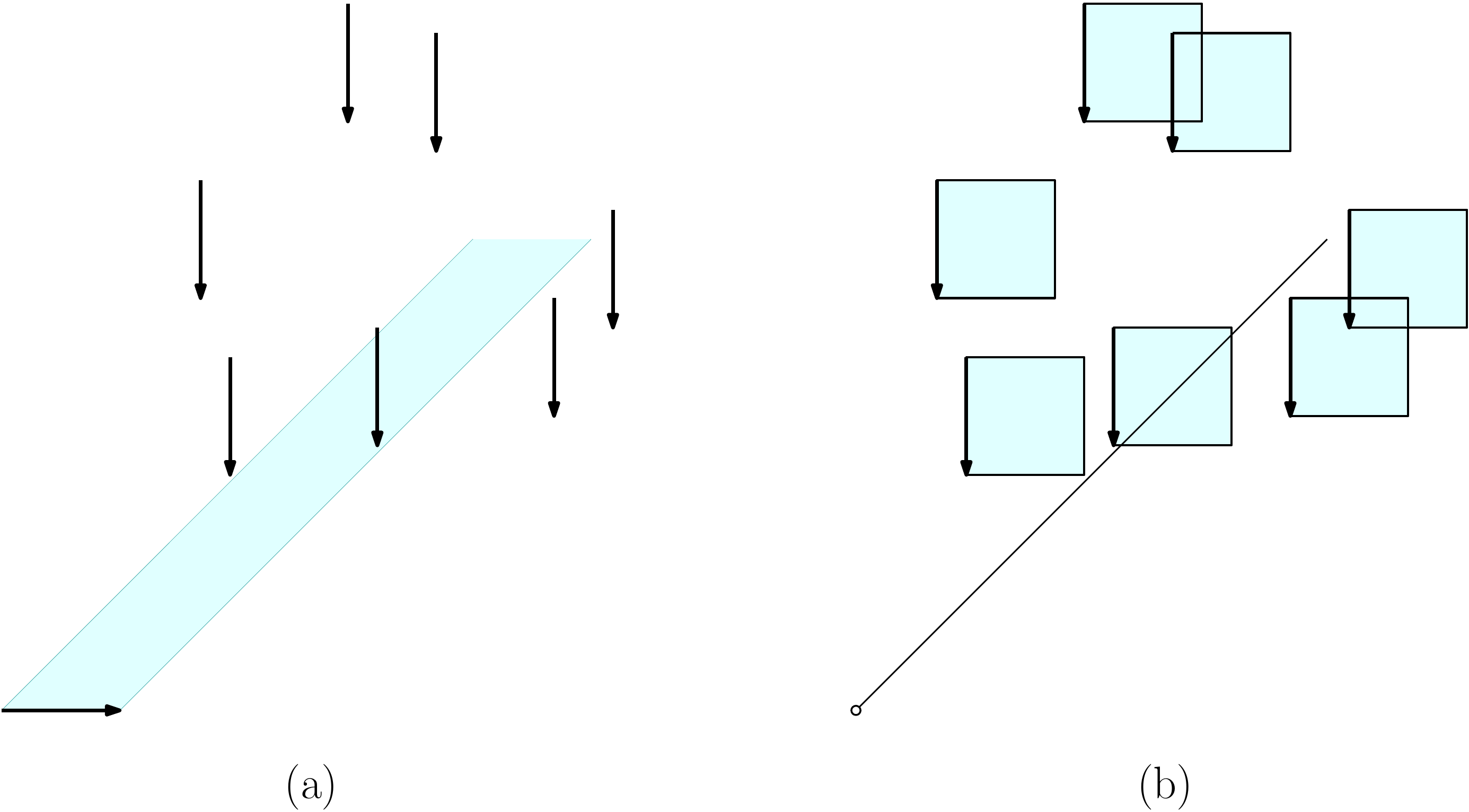}}
\caption{(a) A random traffic crossing problem as viewed from a single active vehicle. (b) The resulting space after the point transformation.}
\label{fig:space_transform}

\vspace*{-10pt}

\end{figure}

Given the global speed limit $\speedlimit$, there are regions in front of each obstacle in which a collision is inevitable (this concept is similar to the obstacle avoidance work done in \cite{petti_safe_2005}).  These triangular zones (referred to as \emph{collision zones}) are based on the speed constraints of the vehicles and are formed by a downward extension of the leading edge of each obstacle.  The leftmost point of this edge is extended vertically and the rightmost point is extended at a slope derived from the ratio between $\speedlimit$ and the obstacle speed.  As one last concession to clarity, we scale the axes of our problem space so that this ratio becomes $1$.  Formally, a collision zone $Z_O$ for the obstacle $O$ is the set of all points $p$, such that there is no path originating at $p$ with a piecewise slope in the interval $[1, \infty]$ that does not intersect $O$.

Expanding the vehicles in $N$ into rectangular obstacles may cause some to overlap, producing larger obstacles and, consequently, larger collision zones.  This merger and generation of collision zones is done through a standard sweep line algorithm and occurs in $O(n \log n)$ steps, where $n$ is the number of obstacles, as described below.

\vspace*{-5pt}

\subsubsection{Merging Obstacles and Growing Collision Zones}
This process is done using a horizontal sweep line moving from top to bottom.  While the following is a relatively standard application of a sweep line algorithm, it is included for the sake of completeness.  First, the event list is populated with the horizontal edges of every obstacle, in top-to-bottom order, requiring $O(n \log n)$ time for $O(n)$ obstacles.  The sweep line status stores a set of intervals representing the interiors of disallowed regions (e.g., the inside of an obstacle or collision zone).  Each interval holds three pieces of information: the location of its left edge, a sorted list of the right edges of any obstacles within the interval, and the slopes of these right edges.  These slopes will be either infinite (i.e., the edges are vertical) or will have a slope of $1$.

In addition to horizontal edge positions, the event list must keep track of three other events which deal with the termination of the sloped edges of the collision zones.  These edges begin at the bottom right edge of an obstacle and terminate in one of three ways: against the top of another obstacle, against the right edge of another obstacle, or by reaching the left edge of an interval.  The first case is already in the event list as the top edges were added at the start of this process.  The remaining two cases are added as the sweep line progresses through the obstacles. 

\begin{figure}[htbp]
\centerline{\includegraphics[scale=.32]{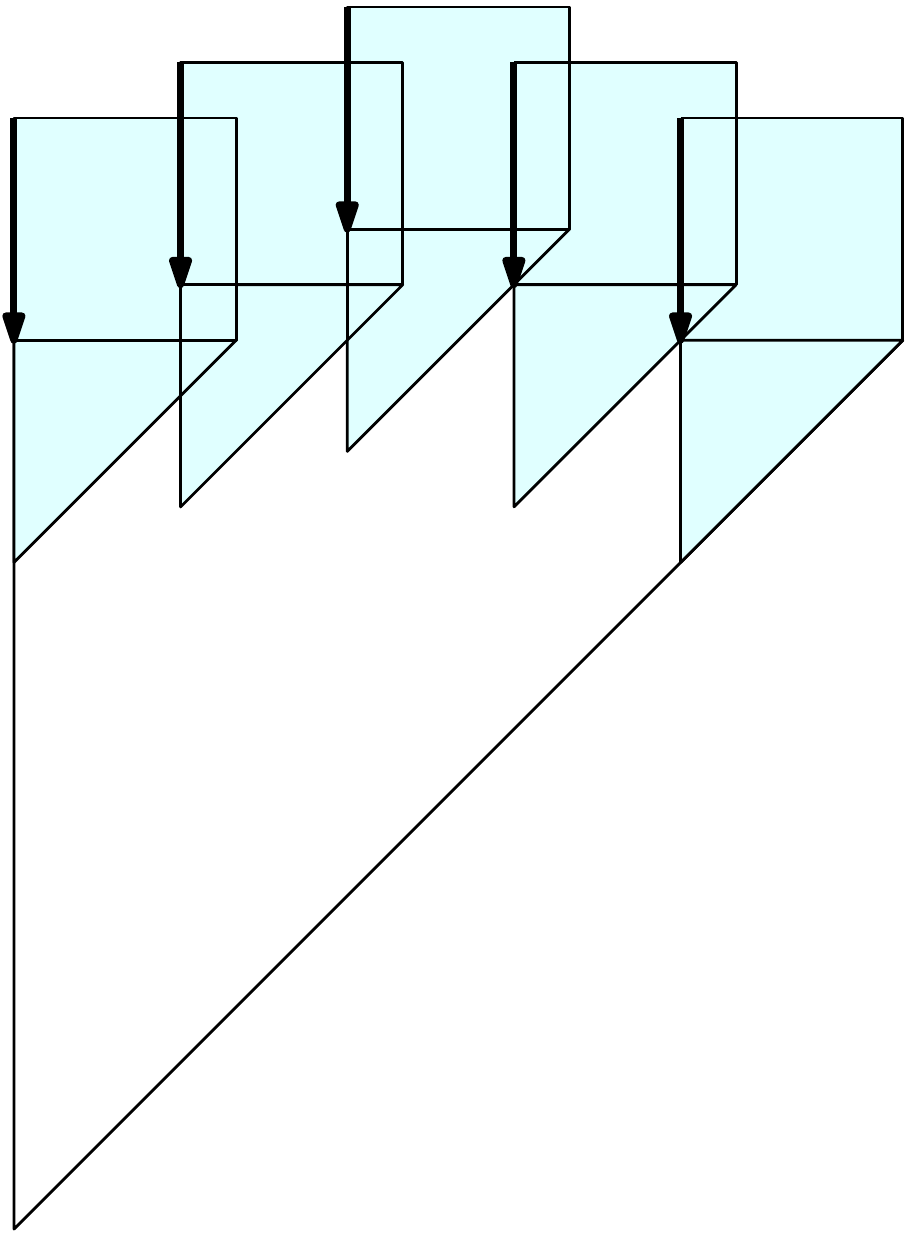}}
\caption{An example illustrating the need to redefine collision zones when obstacles overlap. Here, the collision zones for each individual obstacle (represented as shaded triangles) are insufficient as the merger creates a larger area that vehicles must avoid (seen here as the unfilled triangle).}
\label{fig:obstacle_merge}
\end{figure}

So, when merging obstacles, the sweep line must handle the following events:
\begin{description}
\item[Top Edge Encounter - ] When the sweep line encounters the top edge of an obstacle it must either create a new interval or add this obstacle to an existing interval.  The creation of a new interval is straightforward as the endpoints of the edge are all that need to be added (see Fig.~\ref{fig:sweep_line}(a)). 

If the top edge intersects an existing interval, however, there is a little more work to be done.  First, if the leftmost point of the edge does not lie within the interval then it becomes the new leftmost edge of the interval (see Fig.~\ref{fig:sweep_line}(b)).  If the sloped edge of a collision zone has already formed for this contiguous block of obstacles (see \textbf{Bottom Edge Encounter} for a description of how these form), then the termination point of the sloped edge may need to be updated to account for a shift in the leftmost edge. 

Second, the rightmost point of the encountered edge is inserted into the list of right edges in left-right order.  The new edge may become the new rightmost edge and if the previous rightmost edge was sloped then it is removed from the edge list. For example, in Fig.~\ref{fig:sweep_line}(d) this has just occurred within the set of obstacles on the left.  If the newly added right edge does not replace the sloped edge and the sloped edge intersects the newly added edge, the point at which they intersect is added to the list of events to be processed (this occurs in Fig.~\ref{fig:sweep_line}(c) on the right side).  If there is an existing event in the event list for the sloped edges intersection with another obstacle, it must be deleted as the addition of the newest obstacle will truncate the edge before it reaches that event. 

\item[Bottom Edge Encounter - ]  When the bottom edge of an obstacle is encountered, the obstacle's right edge is found in the interval's edge list.  If it is not the rightmost, it is removed from the edge list (this occurs in Fig.~\ref{fig:sweep_line}(e) on the left, denoted by the grey slope arrow).  If the edge to be removed is the rightmost edge in the list, rather than removing it, its slope is changed to that of the ratio between the vehicles' speed limit and the speed of the vehicles, $\frac{\speedlimit}{s_n}$.  Next, the termination point for this sloped edge is added to the event list.  This is the point at which the leftmost edge of the interval and the sloped edge meet. This point is illustrated in Fig.~\ref{fig:sweep_line}(e), though it was added when the previous bottom edge was processed.  As noted above, this event may need to be updated if a top edge is encountered that moves the leftmost edge of this interval.

\item[Sloped Edge Termination -] When the sloped edge terminates against a right edge, it is deleted from the edge list.  This makes the edge with which it collided the new rightmost edge.

\item[Interval Termination - ] In this case, the sloped edge of the collision zone has met the leftmost edge of the interval.  When this is the case, the interval has finally closed and can be removed from the sweep line status (see Fig.~\ref{fig:sweep_line}(f)).
\end{description}

\begin{figure}[htbp]
\centerline{\includegraphics[scale=.60]{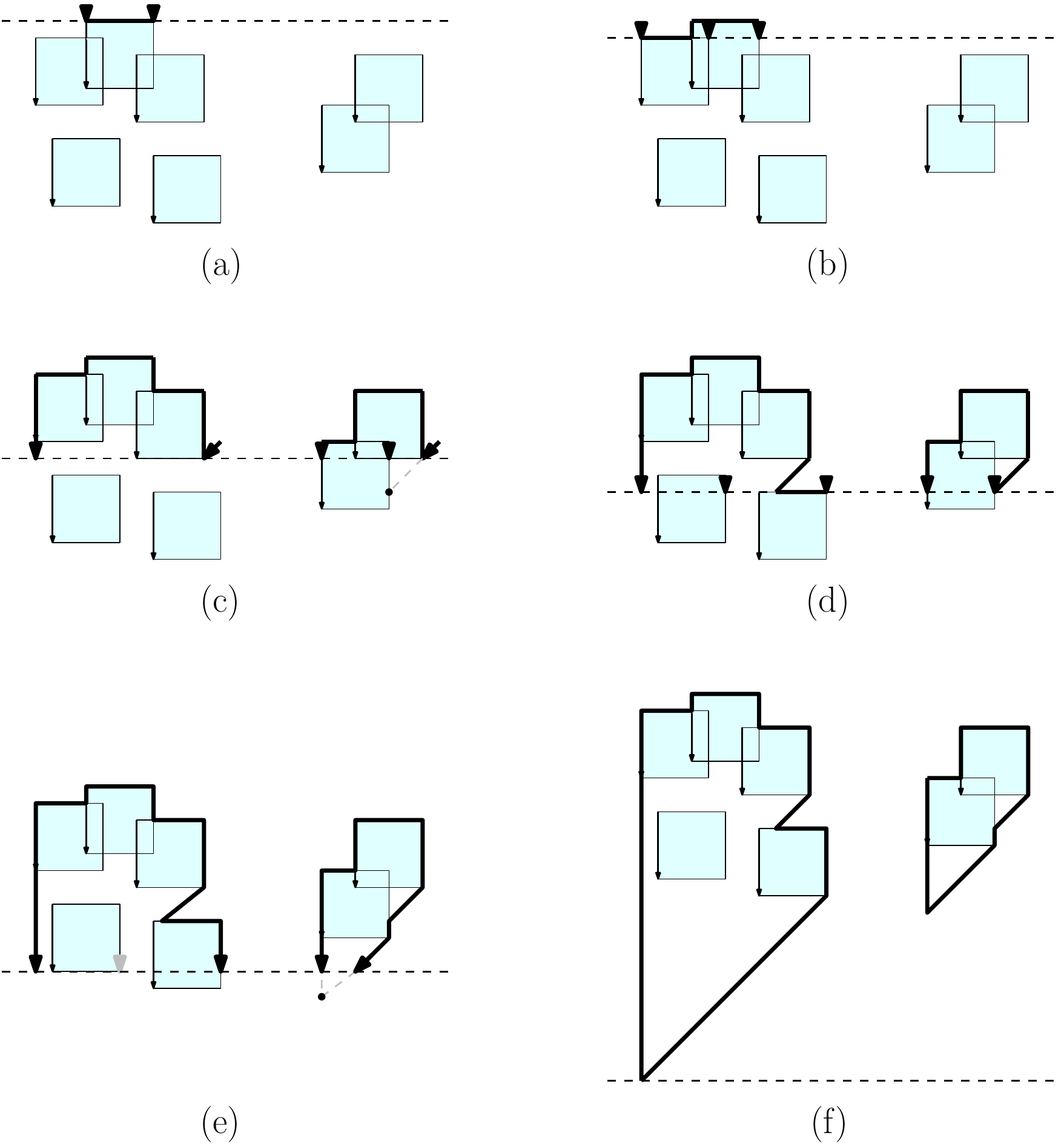}}
\caption{A sweep line merging obstacles and creating collision zones.  Note: these illustrations do not show every step in the sweep line process.  Some are skipped in order to save space.  (a) Encountering the first top edge and adding an interval to the sweep line status.  (b) Encountering the next top edge, which increases the interval size.  (c) Encountering bottom edges changes the rightmost slope of the collision zone.  Notice on the right that an internal right edge is stored in the status. (d) Sloped edges the top of an unprocessed obstacle and the rightmost edge of an obstacle in an interval. (e) Encountering the bottom edge of an internal obstacle.  It's rightmost edge is deleted from the sweep line status. (f) Reaching the point of convergence for a collision zone.  The interval is deleted from the sweep line status.}
\label{fig:sweep_line}
\end{figure}

The initial population of the event list occurs in $O(n \log n)$. As the sweep line progresses through the obstacle space, it adds and removes the right edges of obstacles to the appropriate intervals.  These lists of edges are built incrementally in sorted order, requiring only $O(\log n)$ time.  Finally, as there is a constant number of possible events per obstacle (a single top edge, a single bottom edge, and a single termination of its sloped edge), there are at most $O(n)$ events to be processed. Thus, the sweep line processes the obstacle space in $O(n \log n)$ time.

\vspace*{-10pt}

\subsubsection{Movement Planning}\label{sec:movement_planning}
Once the obstacles have been merged and grown appropriately, we need to find speed profiles for each vehicle that allow them to safely cross the intersection.  This is done using the same obstacle filled space we have been working with thus far, though with a small change in perspective.  Currently, vehicles are only allowed horizontal movement and obstacles only move vertically.  Instead, we will treat the obstacles as static objects and add a vertical component to the vehicles equal to the obstacles' speed.  So, for example, a vehicle moving at the maximum speed will actually follow a path with a slope of $\frac{s_n}{\speedlimit}$ whereas a stationary vehicle will travel vertically.  Again, we have scaled our axes so that this ratio is $1$, imposing on the vehicle monotonic movement with a slope in the interval $[1, \infty]$.  With this understanding, we can now easily find a path through the obstacles while obeying the speed constraints of the vehicles. 

The vehicle will travel at its minimum slope (equivalent to its maximum speed) until it either reaches its goal position or encounters an obstacle.  If an obstacle is encountered, the vehicle travels vertically until it is no longer blocked (this vertical motion corresponds to stopping and waiting for the obstacle to pass).  Once this occurs, the vehicle continues on its way at its maximum speed until it has covered the distance to its goal (measured horizontally, as vertical movement no longer represents spatial translation).

The path created by the above behavior can be efficiently found through the use of another line sweep.  First, notice that every edge that is locally to the left of an obstacle (referred to as a \emph{left edge}) is a vertical line segment.  Since the vehicles move monotonically, they will only ever encounter an obstacle at one of these left edges.  So, to find a path for each vehicle traveling at speed $\speedlimit$, a sweep line perpendicular to the vehicles' trajectories is created and swept from the upper-right to the lower-left (see Fig.~\ref{fig:pathfinding_line_sweep}(a)).  This perpendicular line's status will maintain a list of obstacle occlusions with respect to the vehicles' direction of travel by adding an interval for each obstacle as it is encountered during the sweep.  More specifically, it stores the point where the sweep line first encountered the obstacle's left edge, the horizontal position of the left edge, and the point where the sweep line last encountered the edge.

\begin{figure}[htbp]
\centerline{\includegraphics[scale=.60]{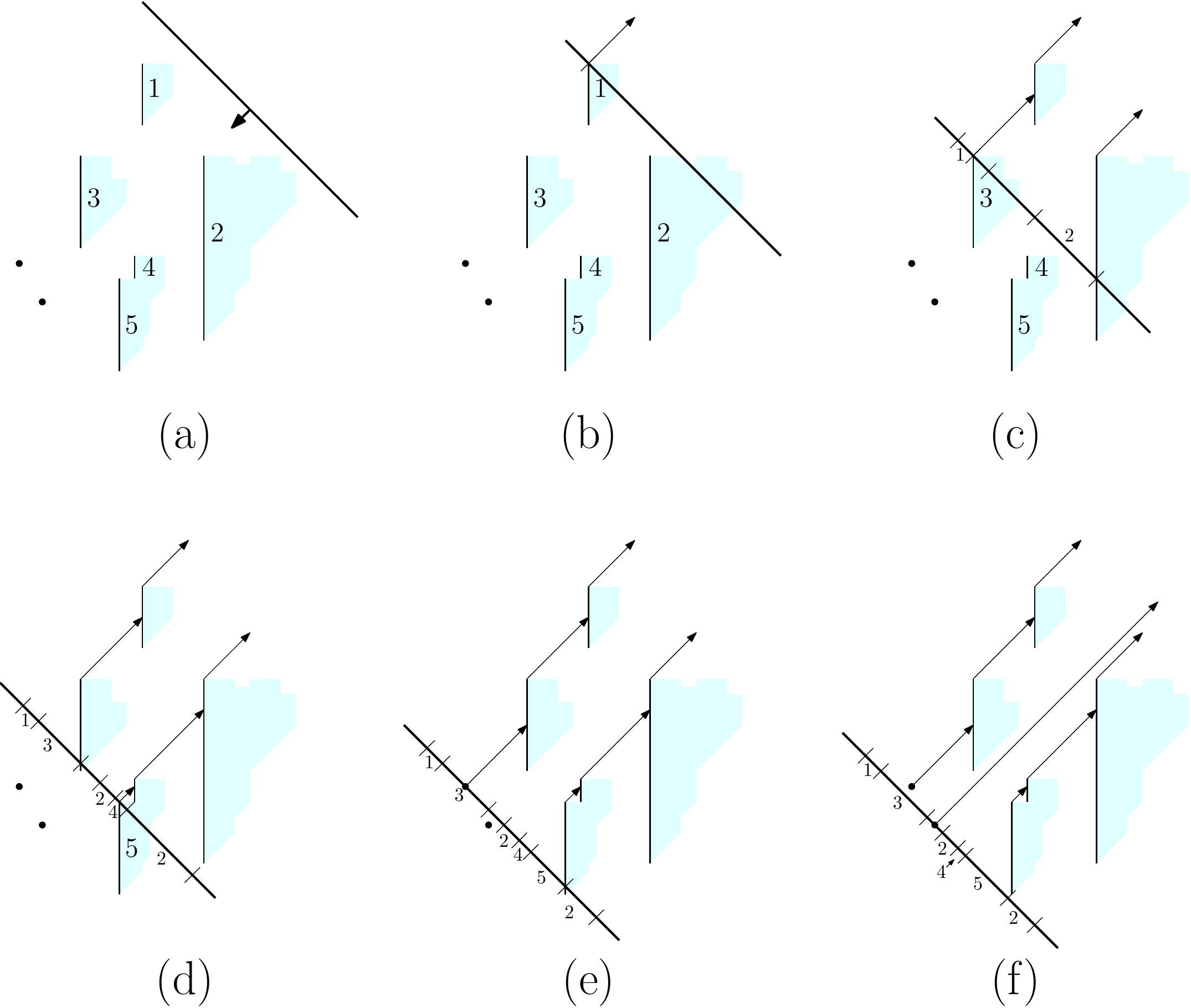}}
\caption{(a) A sweep line for path finding, traveling perpendicular to the direction of travel of a vehicle moving at speed $\speedlimit$. (b) The sweep line encountering vertical edge $1$.  As there is no interval on the sweep line where it occurs, this line's path goes directly to the goal at speed $\speedlimit$.  Edges to the root of the path tree are represented by arrows going off to infinity.  (c) The sweep line encountering edge $3$.  This encounter lies in the interval for edge $1$.  (d)  Encountering edge $5$, creating a path from it to edge $4$.  (e) Encountering the first vehicle, which lies in the interval for edge $3$.  Thus, the final path for the vehicle is to travel at maximum speed until it reaches edge $3$, wait for the edge to pass, travel to $1$, wait, and finally travel to the goal position.  (f) The sweep line encountering the second vehicle at an open interval.  Thus, this vehicle can travel at speed $\speedlimit$ until it reaches its goal position.}
\label{fig:pathfinding_line_sweep}
\end{figure}

During the sweep, a tree is built representing a set of all paths through the obstacle field that encounter an obstacle. Vehicles will either encounter an obstacle in the tree or are free to travel at full speed without collision until their goal is reached.  Each obstacle is a vertex in the tree and edges represent the path taken after encountering this obstacle.  The edge will either lead to an encounter with another obstacle or will lead to the root.  The root is the only vertex which does not represent an obstacle but instead signifies an open path to the goal.

The event list for the sweep line is populated with the upper and lower ends of each left edge.  Whenever an upper end is encountered, it is inserted into the list of intervals in the sweep line status and the obstacle is inserted into the path tree.  If the insertion point does not lie within an existing interval, then an edge between the obstacle and the root is created (see Fig.~\ref{fig:pathfinding_line_sweep}(b)).  If the insertion point lies within an interval, that interval is split by the inserted point and an edge between the new obstacle and the interval's obstacle is added (see Fig.~\ref{fig:pathfinding_line_sweep}(c)).

Whenever the lower end point of an obstacle's left edge is encountered, the interval ending for that obstacle is added to the list.  If an event occurs before an interval has completed, the intervals intermediate size can be determined using the position of the sweep line, the start point of the interval, and the position of the obstacle's left edge (This occurs in Fig.~\ref{fig:pathfinding_line_sweep} between (c) and (d)).  Finally, when a vehicle is encountered, its position along the sweep line determines its path.  If it is an interval, then its path begins by traveling to the associated obstacle and, using the path tree, travels to that obstacle's parent obstacle, repeating this process until it has reached its goal position.

So, in the example in Fig.~\ref{fig:pathfinding_line_sweep}(e), the upper vehicle encounters obstacle $3$, waits for it to pass (i.e., travels vertically till the end is reached), moves at the maximum speed until it encounters obstacle $1$, then continues on until it reaches its goal position.  The lower vehicle, having been inserted into the interval list in between intervals, is free to travel at the maximum speed until its goal position is reached (see Fig.~\ref{fig:pathfinding_line_sweep}(f)).

\vspace*{-10pt}

\subsection{Intersection Between a One-Way Street and a Two-Way Highway}
In this case, vertical vehicles approach from the North and the South while horizontal vehicles travel in a single lane.

To account for the bidirectional vertical vehicles we fold the space along the horizontal lane.  This rotates the northbound traffic to an equivalent southbound set of vehicles (see Fig.~\ref{fig:bidirectional_obstacle_folding}).  This only requires a $O(n)$ transformation.  Using the plane sweep algorithm above yields a combined obstacle space.

\begin{figure}[htbp]
\vspace*{-10pt}

\centerline{\includegraphics[scale=.40]{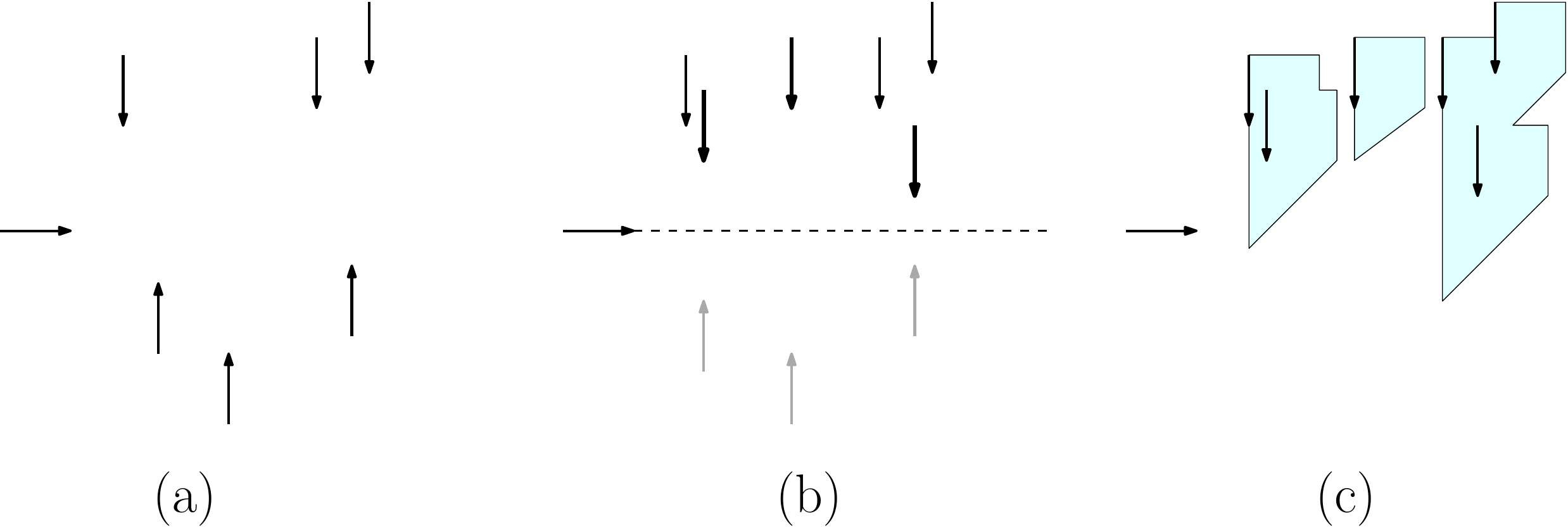}}

\vspace*{-5pt}

\caption{(a) An example of bidirectional cross-traffic. (b) To account for how these vehicles interact when they reach a horizontal lane, we can fold the space along the lane, rotating one set of vehicles about it. (c) Then, we run the same space transformation and obstacle merger detailed above.}
\label{fig:bidirectional_obstacle_folding}

\vspace*{-10pt}

\end{figure}

Finally, we must prevent the vehicles from rear-ending each other.  Once the lead vehicle has found a motion plan through the obstacles, it creates a new set of constraints for the vehicles behind it.  The monotonic path of the lead vehicle is stored in a binary search tree, allowing for easy collision queries.  

To begin with, the lead vehicle's path needs to be added to the tree.  However, because the lead vehicle is represented by a single point coincident with the front of the vehicle, the path we are storing needs to be shifted leftward by an amount equal to the vehicle's length (see Fig.~\ref{fig:same_lane_pathfinding}(b)).  The next vehicle in the lane must not collide with this newly created boundary.  

In the simplest case, the trailing vehicle can simply adopt the same movement policy as prescribed by the algorithm in Section~\ref{sec:movement_planning}.  However, this solution needs to be modified in the following two cases:  (i) when the lead vehicle's shifted path pierces an obstacle, closing off any space through which the trailing vehicle could follow or (ii) when the trailing vehicle would collide with the rear of the lead vehicle.

In case (i), the concern is the creation of overhangs in the obstacle space.  In Section~\ref{sec:one_sided_solution}, overhangs were eliminated through the merger of obstacles.  In this case, however, vertical portions of a lead vehicle's path may pierce an obstacle, recreating such overhangs.  While the merging algorithm could be used to eliminate them once again, doing so for each vehicle is too costly.

Instead, we keep track of which obstacle lies directly below the left edge of each obstacle (this can be found during the obstacle merge plane sweep or through ray shooting).  If an obstacle is pierced by the leading vehicle's path then the space below it can no longer be part of a viable path, as any vehicle entering this space will become trapped.  If another obstacle lies directly below the first, it will need to close off the space below itself as well.  It is possible for this operation to cascade down through multiple obstacles, but once an obstacle has closed off the space below it, it will never need to make this update again.  Thus, this update only requires $O(n)$ time.  This cascading path is then added to the boundary created above, pruning a portion of the search tree and replacing it with the new path (see Fig.~\ref{fig:same_lane_pathfinding}(c)).

For case (ii), the trailing vehicle must check for a collision with the boundary when traveling at full speed (i.e., any time it leaves the upper corner of an obstacle and travels at slope $1$).  Given the binary search tree in which this boundary is stored, we can query for collisions in $O(\log n)$ time for each obstacle the trailing vehicle encounters.  If a collision occurs, the trailing vehicle simply follows the boundary from that point forward.  Because overhangs were eliminated above, this leads to the fastest collision free motion plan for the trailing vehicle.  As each obstacle will be subsumed by the boundary once the trailing vehicle has determined its motion plan, this query will only ever occur once for each obstacle, leading to a time complexity of $O(n \log n)$.

\begin{figure}[htbp]
\centering
\includegraphics[scale=.35]{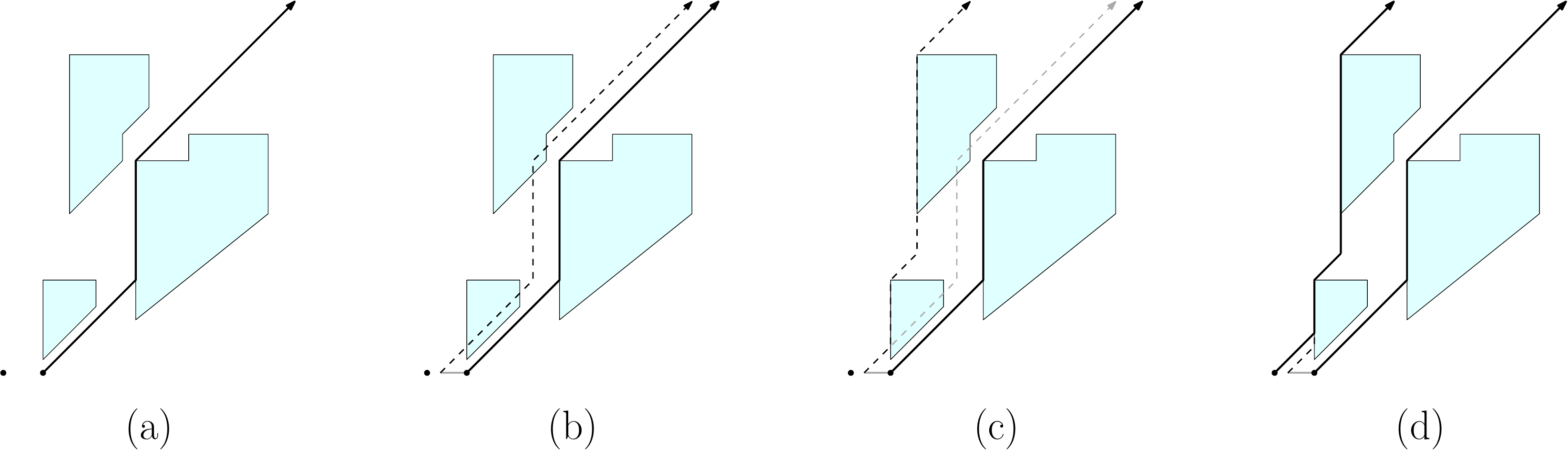}
\caption{An example of two vehicles in the same lane planning their motion through a set of obstacles. (a) As before, the lead vehicle projects its path and queries the first encountered obstacle for the remainder of the motion plan. (b) The path taken by the lead vehicle's rear end is used to create a boundary for the vehicles behind it. (c) This boundary creates overhangs which are eliminated by bending it around the obstacle(s) that create the overhang.  (d) The motion planning algorithm used in Section~\ref{sec:movement_planning} is used here, taking the new boundary into account.}
\label{fig:same_lane_pathfinding}
\end{figure}

In the end, we can still account for shared lanes without a running time greater than $O(n \log n)$.

\vspace*{-10pt}
	
\subsection{Intersection Between Two-Way Highways}
Finally, this case combines the two above, allowing for bidirectional movement horizontally and vertically, with multiple lanes along each axis.  

The vehicles approaching from the East are independent of those approaching from the West, presenting a symmetric problem that can be solved with the techniques discussed above. The addition of horizontal lanes, however, impacts the running time of the algorithm.  Previously, the bidirectional vertical traffic was accounted for by folding the obstacle space along a single horizontal lane, but in this case, because the position of the vertical vehicles relative to each other is different at any given lane, the folding must occur individually for each lane.  Thus, the algorithm runs in $O(kn \log n)$, for $k$ horizontal lanes.  In general, we assume that $k$ is a relatively small constant.

\vspace*{-10pt}

\section{Traffic Crossing in the Discrete Setting}\label{sec:discrete_setting}

\vspace*{-5pt}

In this section we consider the problem in a simple discrete setting, significantly simplifying the description of the algorithms and freeing us from a number of cumbersome continuous issues while still capturing the most salient elements of the original traffic-crossing problem. We assume that each vehicle occupies a point on the integer grid in the plane, $\ZZ^2$. Time advances discretely in unit increments, and at each time step a vehicle may either advance to the next grid point or remain where it is. A collision occurs if two vehicles occupy the same grid point.

The \emph{discrete Traffic Crossing Problem} is defined in much the same manner as in the continuous case. The problem is presented as a set $V$ of $n$ vehicles on the integer grid. Each vehicle $v_i$ is represented by its initial and goal positions $\initpos_i$ and $\goalpos_i$, respectively, both in $\ZZ^2$. Also given are a starting time $\inittime_i$ and deadline $\goaltime_i$, both in $\ZZ^+$ (where $\ZZ^+$ denotes the set of nonnegative integers). A vehicle's direction $d_i$ is a unit length vector directed from its initial position to its goal, which is either horizontal or vertical. Time proceeds in unit increments starting at zero. The motion of $v_i$ is specified as a function of time, $\delta_i(t) \in \{0,1\}$. Setting $\delta_i(t) = 0$ means that at time $t$ vehicle $i$ remains stationary, and $\delta_i(t) = 1$ means that it moves one unit in direction $d_i$. Thus, $v_i$'s position at time $t \ge 0$ is $p_i(t) = \initpos_i + d_i \sum_{x=0}^{t} \delta_i(x)$. 


Generalizing the problem definition from Section~\ref{sec:prob_def}, the objective is to compute a speed profile $D = \ang{\delta_1, \ldots, \delta_n}$ involving all the vehicles that specifies a collision-free motion of the vehicles in such a manner that each vehicle starts at its initial position and moves monotonically towards its goal, arriving there at or before its given deadline. Similar to road networks, we assume that along any horizontal or vertical grid line, the vehicle direction vectors are all the same.

\vspace*{-10pt}

\subsection{Maximum Delay}

\vspace*{-5pt}

Because we will be largely interested in establishing approximation bounds in this section, we will depart from the decision problem and consider a natural optimization problem instead, namely, minimizing the maximum delay experienced by any vehicle, defined formally as follows. For each vehicle we consider only its initial and goal positions, and let us assume that all vehicles share the same starting time at $t = 0$. A vehicle $v_i$ experiences a \emph{delay} at time $t$ if it does not move at this time (that is, $\delta_i(t) = 0$). The \emph{total delay} experienced by a vehicle is the total number of time instances where it experiences a delay until the end of the motion simulation. The \emph{maximum delay} of the system is the maximum total delay experienced by any vehicle. 

While we will omit a formal proof, it is not hard to demonstrate that the NP-hardness reduction of Section~\ref{sec:hardness} can be transformed to one showing that it is NP-hard to minimize maximum delay in the discrete setting. (Intuitively, the reason is that the reduction involves purely discrete quantities: integer vehicle coordinates and starting times, vehicles of unit length, and unit speed limit. The system described in the reduction is feasible if and only if the maximum delay is at most five time units.) However, it is interesting to note that the question of whether there exists a solution involving at most single unit delay can be solved efficiently. This is stated in the following result.

\begin{theorem} \label{thm:max-delay-one}
There exists an $O(n m)$ time algorithm that, given an instance of the discrete Traffic Crossing Problem with $n$ vehicles where each vehicle encounters at most $m$ intersections, determines whether there exists a solution with maximum delay of at most one time unit.
\end{theorem}

\begin{proof}
We will show that the problem, can be reduced to an instance of 2-SAT in $O(n m)$ time. The result follows from the fact that 2-SAT can be solved in linear time \cite{aspvall:1979}. We use the easy observations that clauses of the form $(x \Rightarrow y)$ (implies) and $(x \oplus y)$ (exclusive-or) can both be expressed in 2-SAT form. 

Since the maximum delay is one unit, throughout the motion process each vehicle may be in one of two states, either having not experienced any delay up to that point or having experienced a single delay. For each vehicle $v_i$ and each intersection $k$ it passes through, we create a boolean variable $x_{i,k}$, whose value will be {\TRUE} to signify that this vehicle has experienced a delay on entry to intersection $k$, and otherwise its value is {\FALSE}.

We generate a 2-SAT instance with the following clauses, which together enforce the conditions of a solution with a maximum delay of one unit:
\begin{enumerate}
\item[(i)] For each vehicle $v_i$ and each pair of consecutive intersections $k$ and $k'$ that $v_i$ passes through, if $v_i$ is delayed on entering $k$, then it is still delayed on entering $k'$. Add the clause $(x_{i, k} \Rightarrow x_{i, k'})$.

\item[(ii)] For each intersection $k$, if two vehicles $v_i$ and $v_j$, one horizontal and one vertical, pass through $k$, and their starting positions are equidistant from $k$, they cannot both be in the same state when arriving at this intersection, for otherwise they would collide. Add the clause $(x_{i, k} \oplus x_{j, k})$.

\item[(iii)] For each intersection $k$, if two vehicles $v_i$ and $v_j$, one horizontal and one vertical, pass through $k$, and $v_i$ is one unit closer to $k$ than $v_j$, if $v_i$ delays on entering $k$, then $v_j$ must delay as well to avoid a collision. Add the clause $(x_{i, k} \Rightarrow x_{j, k})$.

\item[(iv)] For each pair of vehicles $v_i$ and $v_j$ in the same lane such that $v_i$'s initial position is one unit before $v_j$'s initial position, if $v_i$ delays then $v_j$ must delay as well to avoid rear-ending $v_i$. For all intersections $k$ through which these vehicles pass, add the clause $(x_{i,k} \Rightarrow x_{j,k})$.
\end{enumerate}

Because only one-unit delays are tolerated, all instances of potentially colliding vehicles are handled by cases~(ii), (iii), and~(iv). It is easy to verify that the resulting formula is satisfiable if and only if there exists a solution to the Traffic Crossing Problem with maximum delay of at most one time unit. In particular, each vehicle $v_i$ travels without delay until (if ever) encountering the first intersection $k$ such that $x_{i,k} = {\TRUE}$, at which time it delays one time unit. Conversely, in any valid solution, each vehicle $v_i$ will delay at most once, which is modeled by setting $x_{i,k} = {\TRUE}$ for all subsequent intersections $k$. Clauses of type~(i) prevent delayed vehicles from illegally changing state, clauses of types~(ii) and~(iii) prevent collisions from perpendicular lanes, and clauses of type~(iv) prevent collisions from vehicles in the same lane. The formula clearly involves $O(n m)$ variables and $O(n m)$ clauses and therefore the reduction runs in $O(n m)$ time.
\end{proof}

\vspace*{-10pt}

\subsection{The Parity Heuristic}

In the discrete setting it is possible to describe a simple common-sense heuristic. Intuitively, each intersection will alternate in allowing horizontal and vertical traffic to pass. Such a strategy might be far from optimal because each time a vehicle arrives at an intersection, it might suffer one more unit of delay. To address this, whenever a delay is imminent, we will choose which vehicle to delay in a manner that will avoid cross traffic at all future intersections. Define the \emph{parity} of a grid point $p = (p_x, p_y)$ to be $(p_x + p_y) \bmod 2$. Given a horizontally moving vehicle $v_i$ and a time $t$, we say that $v_i$ is \emph{on-parity} at $t$ if the parity of its position at time $t$ equals $t \bmod 2$. Otherwise, it is \emph{off-parity}. Vertically moving vehicles are just the opposite, being \emph{on-parity} if the parity of their position is \emph{not} equal to $t \bmod 2$. Observe that if two vehicles arrive at an intersection at the same time, one moving vertically and one horizontally, exactly one of them is on-parity. This vehicle is given the right of way, as summarized below.

\begin{description}
\item[Parity Heuristic:] If two vehicles are about to arrive at the same intersection at the same time $t$, the vehicle that is on-parity proceeds, and the other vehicle waits one time unit (after which it will be on-parity, and will proceed).
\end{description}

The parity heuristic has a number of appealing properties. First, once all the vehicles in the system are on-parity, every vehicle may proceed at full speed without the possibility of further collisions. Second, the heuristic is not (locally) wasteful in the sense that it does not introduce a delay into the system unless a collision is imminent. Finally, the rule is scalable to large traffic systems, since a traffic controller at an intersection need only know the current time and the vehicles that are about to enter the intersection. 

\vspace*{-10pt}

\subsection{Steady-State Analysis of The Parity Heuristic}

Delays may be much larger than a single time unit under the parity heuristic. (For example, a sequence of $k$ consecutive vehicles traveling horizontally that encounters a similar sequence of $k$ vertical vehicles will result in a cascade of delays, spreading each into an alternating sequence of length $2 k$.%
%
) This is not surprising given the very simple nature of the heuristic. It is not difficult to construct counterexamples in which the maximum delay of the parity heuristic is arbitrarily large relative to an optimal solution. We will show, however, that the parity heuristic is asymptotically optimal in a uniform, steady-state scenario (to be made precise below).

Consider a traffic crossing pattern on the grid. Let $m_x$ and $m_y$ denote the numbers of vertical and horizontal lanes, respectively. Each lane is assigned a direction arbitrarily (up or down for vertical lanes and left or right for horizontal). Let $R$ denote a $W \times W$ square region of the grid containing all the intersections (see Fig.~\ref{fig:parity_analysis}(a)). In order to study the behavior of the system in steady-state, we will imagine that $R$ is embedded on a torus, so that vehicles that leave $R$ on one side reappear instantly in the same lane on the other side (see Fig.~\ref{fig:parity_analysis}(b)). Equivalently, we can think of this as a system of infinite size by tiling the plane with identical copies (see Fig.~\ref{fig:parity_analysis}(c)). 
We assume that $W$ is even.

\begin{figure}[htbp]

\vspace*{-10pt}

\centerline{\includegraphics[scale=.40]{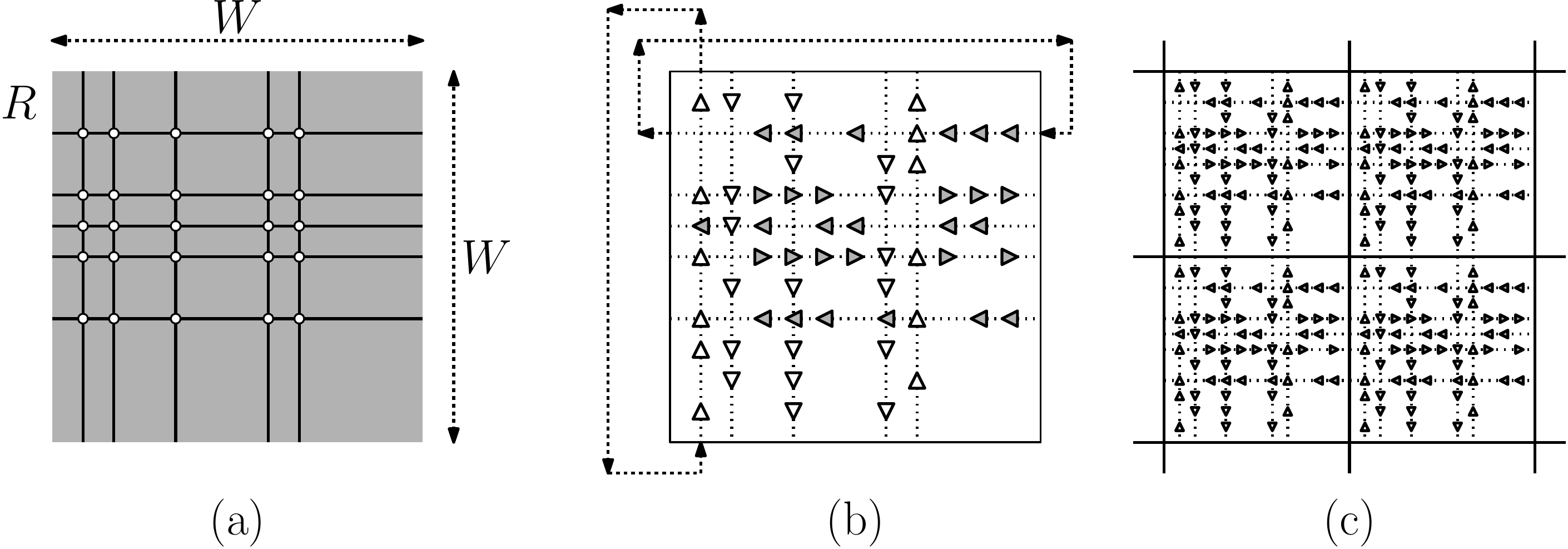}}
\caption{Analysis of the Parity Heuristic.}
\label{fig:parity_analysis}

\vspace*{-10pt}

\end{figure}

If the system is sufficiently dense, the maximum delay of the system will generally grow as a function of time. Given a scheduling algorithm and a discrete traffic crossing, define its \emph{delay rate} to be the maximum delay after $t$ time units divided by $t$. Define the \emph{asymptotic delay rate} to be the limit supremum of the delay rate for $t \rightarrow \infty$. Our objective is to show that, given a suitably uniform traffic crossing instance on the torus, the asymptotic delay rate of the parity algorithm is optimal.

We say that a traffic crossing on the torus is \emph{uniform} if every lane (within the square $R$) has an equal number of vehicles traveling on this lane. Letting $n'$ denote this quantity, the total number of vehicles in the system is $n = n'(m_x + m_y)$. (The total number of positions possible is $W(m_x + m_y) - m _x m_y$, and so $n' \le W - m_x m_y/(m_x + m_y)$.) The initial positions of the vehicles within each of the lanes is arbitrary. Let $p = n'/W$ denote the density of vehicles within each lane. Let $\rpar = \rpar(W,p,m_x,m_y)$ denote the worst-case asymptotic delay rate of the parity heuristic on any uniform discrete traffic crossing instance of the form described above, and let $\ropt = \ropt(W,p,m_x,m_y)$ denote the worst-case asymptotic delay for an optimum scheduler.

Our approach will be to relate the asymptotic performance of parity and the optimum to a parameter that describes the inherent denseness of the system. Define $\chi = \max(0, 2 p - 1)$ to be the \emph{congestion} of the system. Observe that $0 \le \chi \le 1$, where $\chi = 0$ means that the density is at most $1/2$ and $\chi = 1$ corresponds to placing vehicles at every available point on every lane (which is not really possible given that $n' < W$). To demonstrate that the parity heuristic is asymptotically optimal in this setting, it can be shown that $\rpar \le \chi/(1+\chi) \le \ropt$. This is a consequence of the following two lemmas, whose proofs are given below.

\begin{lemma} \label{parity-delay.lem}
Given any uniform traffic crossing instance on the torus with congestion $\chi$, $\rpar \le \chi/(1+\chi)$.
\end{lemma}

\begin{proof}
Consider intersection $q$ in the system and any lane passing through $q$. By unwrapping the torus within the plane, we can think of the vehicles moving towards $q$ as an infinite sequence of blocks that repeats with period $W$. For any $k \ge 1$, consider the vehicles in this lane whose initial positions are within distance $k W$ of $q$ and are directed towards $q$. There are $k n' = k W p$ such vehicles, organized into identical blocks of length $W$. For $1 \le j \le k n'$, let $x_j$ denote the distance from the initial position of the $j$th vehicle to $q$. 

To obtain an upper bound on the delay, we will allow each of these vehicles to pass through $q$ only if it is on-parity. (The parity algorithm may allow off-parity vehicles to pass through if there is no imminent collision, so our assumption results in the highest possible delay.) Let $t_j$ denote the time at which vehicle $j$ passes through $q$ according to the parity heuristic. 

We first establish the following bound on $t_j$:
\[
	t_j
		~ \le ~ \max_{1 \le i \le j} (x_i + 1 + 2(j-i)).
\]
This follows by a simple induction argument. If $j = 1$, then we have $t_j \le x_j + 1$, which matches the time for this vehicle to reach $q$ together with one optional unit of delay if it is out of parity. Otherwise, observe that if the delays prior to vehicle $j$ do not affect it, we have $t_j \le x_j + 1$. If they do, then vehicle $j$ will pass through $q$ two units after $t_{j-1}$, thus yielding
\begin{eqnarray*}
	t_j
		& \le & \max(x_j + 1, t_{j-1} + 2) \\
		& \le & \max(x_j + 1, \max_{1 \le i \le j-1} (x_i + 1 + 2(j-1-i)) + 2) \\
		&  =  & \max(x_j + 1, \max_{1 \le i \le j-1} (x_i + 1 + 2(j-i)))
		~  =  ~ \max_{1 \le i \le j} (x_i + 1 + 2(j-i)).
\end{eqnarray*}
While the maximum is taken over $j$ choices of $i$, we can simplify this due to the periodic nature of the system. Suppose that $i$ is the index that achieves the maximum in the definition of $t_j$. Let $i'$ denote the index of the corresponding vehicle in a block that is closer to $q$. That is, $i = i' + k' n'$ and $x_i = x_{i'} + k' W$, for some $k' \ge 1$. We have
\begin{eqnarray*}
	x_i + 1 + 2(j-i)
		& = & (x_{i'} + k' W) + 1 + 2(j - (i' + k' n')) \\
		& = & x_{i'} + 1 + 2(j-i') - k'(2n' - W) \\
		& = & (x_{i'} + 1 + 2(j-i')) - k' W (2 p - 1).
\end{eqnarray*}

Therefore, if $\chi = 0$ (meaning that $2 p - 1 \le 0$), 
we may assume that $i$ achieves its maximum value among the $n'$ vehicles immediately preceding $j$, that is $j - i \le n'$. (Since using values of $i'$ from earlier blocks can only decrease the value in the max.) Conversely, if $\chi > 0$ (meaning that $2 p - 1 > 0$), $i$ achieves its maximum among the $n'$ vehicles that are closest to $q$. Let $k' = \ceil{j/n'}$ denote the index of the block that contains $x_j$. We conclude that if $\chi = 0$, then since $x_i \le x_j$, we have
\[
	t_j 
		~ \le ~ x_j + 1 + 2 n'
		~ \le ~ k' W + 1 + 2 n'
		~ \le ~ k' W + 2 W.
\]
On the other hand, if $\chi > 0$, then $x_i \le W$, $j - i \le (k'-1) n'$ implying that
\[
	t_j
		~ \le ~ W + 1 + 2 k' n'
		~ \le ~ k' (2 n') + 2 W
		~  =  ~ k' W (2 p) + 2 W.
\]
Note that $2 p > 1$ if and only if $\chi > 0$. Therefore, we conclude that
\[
	t_j
		~ \le ~ k' W \cdot \max(1, 2 p) + 3 W
		~ \le ~ k' W (1 + \chi) + 3 W.
\]
Because vehicle $j$ is at distance $x_j$ from $q$ and every time unit beyond $x_j$ contributes its the delay, the delay is $t_j - x_j$. The delay rate is $(t_j - x_j)/t_j = 1 - x_j/t_j$. Since $x_j \ge k' W$, the delay rate for each vehicle $x_j$ is
\[
	1 - \frac{x_j}{t_j}
		~ \le ~ 1 - \frac{k' W}{k' W (1 + \chi) + 3 W}
		~  =  ~ 1 - \frac{1}{(1 + \chi) + 3/k'}.
\]
For the asymptotic delay rate we are interested in the limit as $k' \rightarrow \infty$, which is at most
\[
	1 - \frac{1}{1 + \chi}
		~ = ~ \frac{\chi}{1 + \chi},
\]
as desired.
\end{proof}

\begin{lemma} \label{opt-delay.lem}
Given any uniform traffic crossing instance on the torus with congestion $\chi$, $\ropt \ge \chi/(1+\chi)$.
\end{lemma}

\begin{proof}
Let $q$ be any intersection of the system. By unwrapping the torus within the plane, we can think of the vehicles moving towards $q$ as an infinite sequence of blocks that repeats with period $W$. For any $k \ge 1$, consider the vehicles on the two lanes incident to $q$ whose initial positions are within distance $k W$. There are $k n' = k W p$ vehicles in each lane. At most one vehicle can pass through $q$ at any time, therefore the total time $T$ for all the vehicles to travel through $q$ is at least $2 k n' = 2 k W p$.

As observed in the proof of Lemma~\ref{parity-delay.lem}, the delay experienced by any vehicle is the difference of its time to arrive at $q$ minus its distance from $q$. Therefore, at least one vehicle among this set (in particular, the last vehicle in one of the two lanes), experiences a delay of at least $\max(0, T - x)$, where $x$ is its distance from $q$. Since $x \le k W$, the maximum delay is at least
\[
	\max(0, 2 k W p - k W) 
		~ = ~ \max(0, k W (2 p - 1)) 
		~ = ~ k W \chi.
\]
By Lemma~\ref{parity-delay.lem} (where $j$ plays the role of the last vehicle in block $k$ to make it through $q$) we have $T \le k W (1 + \chi) + 3 W$ for the parity algorithm, and therefore it can be no higher for the optimum. The delay rate is the delay divided by the time, which is at least
\[
	\frac{k W \chi}{k W (1 + \chi) + 3 W}
		~ = ~ \frac{\chi}{1 + \chi + 3/k}.
\]
To obtain the asymptotic delay rate, we consider the limit as $k$ tends to infinity, which yields a lower bound on the asymptotic delay rate of at least $\chi/(1+\chi)$, as desired.
\end{proof}

While the proofs are somewhat technical, the intuition behind them is relatively straightforward. If $\chi = 0$, then while local delays may occur, there is sufficient capacity in the system for them to dissipate over time, and hence the asymptotic delay rate tends to zero as well. On the other hand, if $\chi > 0$, then due to uniformity and the cyclic nature of the system, delays will and must grow at a predictable rate. As an immediate consequence of the above lemmas, we have the following main result of this section.

\begin{theorem}
Given a uniform traffic crossing instance on the torus, the asymptotic delay rate of the parity heuristic is optimal.
\end{theorem}

\bibliographystyle{plain}
\bibliography{shortcuts,bibliography-squished}

\begin{thebibliography}{10}

\bibitem{Arkin:2008}
E.~M. Arkin, J.~S.~B. Mitchell, and V.~Polishchuk.
\newblock Maximum thick paths in static and dynamic environments.
\newblock {\em Comput.\ Geom.\ Theory Appl.}, 43(3):279--294, April 2010.

\bibitem{aspvall:1979}
B.~Aspvall, M.~F. Plass, and R.~Endre Tarjan.
\newblock A linear-time algorithm for testing the truth of certain quantified
  boolean formulas.
\newblock {\em Inform.\ Process.\ Lett.}, pages 121--123, 1979.

\bibitem{au_motion_2010}
T.-C. Au and P.~Stone.
\newblock Motion planning algorithms for autonomous intersection management.
\newblock In {\em Bridging the Gap Between Task and Motion Planning}, 2010.

\bibitem{berger_travellers_2010}
F.~Berger and R.~Klein.
\newblock A traveller's problem.
\newblock In {\em Proc.\ 26th Annu.\ Sympos.\ Comput.\ Geom.}, {SoCG} '10,
  pages 176--182, New York, {NY}, {USA}, 2010. {ACM}.

\bibitem{carlino_auction-based_2013}
D.~Carlino, S.~D. Boyles, and P.~Stone.
\newblock Auction-based autonomous intersection management.
\newblock In {\em Intelligent Transportation Systems-({ITSC}), 2013 16th
  International {IEEE} Conference on}, pages 529--534. {IEEE}, 2013.

\bibitem{clarke_scheduling_1964}
G.~Clarke and J.~W. Wright.
\newblock Scheduling of vehicles from a central depot to a number of delivery
  points.
\newblock {\em Operations Res.}, 12(4):568--581, August 1964.

\bibitem{dantzig_truck_1959}
G.~B. Dantzig and J.~H. Ramser.
\newblock The truck dispatching problem.
\newblock {\em Management Sci.}, 6(1):80--91, October 1959.

\bibitem{dresner_multiagent_2004}
K.~Dresner and P.~Stone.
\newblock Multiagent traffic management: A reservation-based intersection
  control mechanism.
\newblock In {\em Proc.\ Third Internat. Joint Conf. on Auton.\ Agents and
  Multi.\ Agent Syst.}, pages 530--537. {IEEE} Computer Society, 2004.

\bibitem{dresner_multiagent_2005}
K.~Dresner and P.~Stone.
\newblock Multiagent traffic management: An improved intersection control
  mechanism.
\newblock In {\em Proc.\ Fourth Internat. Joint Conf. on Auton.\ Agents and
  Multi.\ Agent Syst.}, pages 471--477. {ACM}, 2005.

\bibitem{dresner_multiagent_2008}
K.~M. Dresner and P.~Stone.
\newblock A multiagent approach to autonomous intersection management.
\newblock {\em J. Artif. Intell. Res.}, 31:591--656, 2008.

\bibitem{fenton_steering_1976}
R.~E. Fenton, G.~C. Melocik, and K.~W. Olson.
\newblock On the steering of automated vehicles: Theory and experiment.
\newblock {\em IEEE Trans.\ Autom. Control}, 21(3):306--315, June 1976.

\bibitem{fiorini_motion_1998}
P.~Fiorini and Z.~Shiller.
\newblock Motion planning in dynamic environments using velocity obstacles.
\newblock {\em Internat.\ J.\ Robot. Res.}, 17(7):760--772, July 1998.

\bibitem{hearn_sliding_2005}
R.~A. Hearn and E.~D. Demaine.
\newblock Pspace-completeness of sliding-block puzzles and other problems
  through the nondeterministic constraint logic model of computation.
\newblock {\em Theo.\ Comp.\ Sci.}, 343(1–2):72--–96, 2005.

\bibitem{petti_safe_2005}
S.~Petti and T.~Fraichard.
\newblock Safe motion planning in dynamic environments.
\newblock In {\em 2005 {IEEE/RSJ} International Conference on Intelligent
  Robots and Systems, 2005. ({IROS} 2005)}, pages 2210--2215, August 2005.

\bibitem{rajamani_vehicle_2011}
R.~Rajamani.
\newblock {\em Vehicle Dynamics and Control}.
\newblock Springer Science \& Business Media, December 2011.

\bibitem{solomon_algorithms_1987}
M.~M. Solomon.
\newblock Algorithms for the vehicle routing and scheduling problems with time
  window constraints.
\newblock {\em Operations Res.}, 35(2):254--265, March 1987.

\bibitem{vanmiddlesworth_replacing_2008}
M.~Van~Middlesworth, K.~Dresner, and P.~Stone.
\newblock Replacing the stop sign: Unmanaged intersection control for
  autonomous vehicles.
\newblock In {\em Proc.\ Seventh Internat. Joint Conf. on Auton.\ Agents and
  Multi.\ Agent Syst.}, pages 1413--1416. International Foundation for
  Autonomous Agents and Multiagent Systems, 2008.

\bibitem{wurman_coordinating_2008}
P.~R. Wurman, R.~D'Andrea, and M.~Mountz.
\newblock Coordinating hundreds of cooperative, autonomous vehicles in
  warehouses.
\newblock {\em The {AI} magazine}, 29(1):9--19, 2008.

\bibitem{yu_multi-agent_2012}
J.~Yu and S.~M. LaValle.
\newblock Multi-agent path planning and network flow.
\newblock {\em {arXiv}:1204.5717 [cs]}, April 2012.
\newblock {arXiv}: 1204.5717.

\end{thebibliography}

\end{document}